\documentclass[11pt,notitlepage,tightenlines,nofootinbib,superscriptaddress]{revtex4-1}
\pdfoutput=1
\usepackage[utf8]{inputenc}
\usepackage{amsthm}
\usepackage{amssymb}
\usepackage{amsmath}

\usepackage[colorlinks=true, linkcolor=MidnightBlue, urlcolor=MidnightBlue, citecolor=MidnightBlue, linktoc=section, pdftitle={On a gap in the proof of the generalised quantum Stein's lemma and its consequences for the reversibility of quantum resources}]{hyperref}

\usepackage{newpxtext,newpxmath}

\usepackage{bbold}
\usepackage{bbm}
\usepackage{nonfloat}
\usepackage{braket}
\usepackage{dsfont}
\usepackage{mathdots}
\usepackage{mathtools}
\usepackage{enumerate}
\usepackage{csquotes}
\usepackage{stmaryrd}
\usepackage[cal=boondox]{mathalfa}
\usepackage{graphicx}
\usepackage{stackengine}
\usepackage{scalerel}
\usepackage[dvipsnames]{xcolor}
\usepackage{array,xpatch}
\usepackage{makecell}
\newcolumntype{x}[1]{>{\centering\arraybackslash}p{#1}}
\usepackage{tikz}
\usepackage{pgfplots}
\usetikzlibrary{shapes.geometric, shapes.misc, positioning, arrows, decorations.pathreplacing, decorations.text, angles, quotes, backgrounds, calc}
\usepackage{booktabs}
\usepackage{xfrac}
\usepackage{siunitx}
\usepackage{centernot}
\usepackage{comment}
\usepackage{chngcntr}

\newtheorem{thm}{Theorem}
\newtheorem*{thm*}{Theorem}
\newtheorem{prop}[thm]{Proposition}
\newtheorem*{prop*}{Proposition}
\newtheorem{lemma}[thm]{Lemma}
\newtheorem*{lemma*}{Lemma}
\newtheorem{cor}[thm]{Corollary}
\newtheorem*{cor*}{Corollary}
\newtheorem{cj}[thm]{Conjecture}
\newtheorem*{cj*}{Conjecture}

\newtheorem*{Def*}{Definition}

\makeatletter
\def\thmhead@plain#1#2#3{%
  \thmname{#1}\thmnumber{\@ifnotempty{#1}{ }\@upn{#2}}%
  \thmnote{ {\the\thm@notefont#3}}}
\let\thmhead\thmhead@plain
\makeatother

\theoremstyle{definition}
\newtheorem*{rem}{Remark}

\newcommand{\bb}{\begin{equation}\begin{aligned}}
\newcommand{\bbb}{\begin{equation*}\begin{aligned}}
\newcommand{\ee}{\end{aligned}\end{equation}}
\newcommand{\eee}{\end{aligned}\end{equation*}}

\let\textleq\relax
\let\textgeq\relax
\let\texteq\relax

\newcommand{\texteq}[1]{\stackrel{\mathclap{\scriptsize \mbox{#1}}}{=}}
\newcommand{\textleq}[1]{\stackrel{\mathclap{\scriptsize \mbox{#1}}}{\leq}}

\newcommand{\textgeq}[1]{\stackrel{\mathclap{\scriptsize \mbox{#1}}}{\geq}}
\newcommand{\ketbra}[1]{\ket{#1}\!\!\bra{#1}}
\newcommand{\ketbraa}[2]{\ket{#1}\!\!\bra{#2}}

\newcommand{\ve}{\varepsilon}

\newcommand{\id}{\mathds{1}}
\newcommand{\R}{\mathds{R}}
\newcommand{\N}{\mathds{N}}
\newcommand{\C}{\mathds{C}}

\newcommand{\cptp}{\mathrm{CPTP}}

\newcommand{\LR}{L\!R}

\DeclareMathOperator{\Tr}{Tr}

\DeclareMathOperator{\co}{conv}

\DeclareMathAlphabet{\pazocal}{OMS}{zplm}{m}{n}

\DeclareMathOperator{\supp}{supp}
\DeclareMathOperator{\spec}{spec}

\newcommand{\HH}{\pazocal{H}}

\newcommand{\D}{\pazocal{D}}

\newcommand{\MM}{\pazocal{M}}

\newcommand{\II}{\pazocal{I}}
\newcommand{\LL}{\pazocal{L}}

\newcommand{\lsmatrix}{\left(\begin{smallmatrix}}
\newcommand{\rsmatrix}{\end{smallmatrix}\right)}

\stackMath

\stackMath

\makeatletter
\newcommand*\rel@kern[1]{\kern#1\dimexpr\macc@kerna}
\newcommand*\widebar[1]{%
  \begingroup
  \def\mathaccent##1##2{%
    \rel@kern{0.8}%
    \overline{\rel@kern{-0.8}\macc@nucleus\rel@kern{0.2}}%
    \rel@kern{-0.2}%
  }%
  \macc@depth\@ne
  \let\math@bgroup\@empty \let\math@egroup\macc@set@skewchar
  \mathsurround\z@ \frozen@everymath{\mathgroup\macc@group\relax}%
  \macc@set@skewchar\relax
  \let\mathaccentV\macc@nested@a
  \macc@nested@a\relax111{#1}%
  \endgroup
}

\definecolor{Blues5seq1}{RGB}{239,243,255}
\definecolor{Blues5seq2}{RGB}{189,215,231}
\definecolor{Blues5seq3}{RGB}{107,174,214}
\definecolor{Blues5seq4}{RGB}{49,130,189}
\definecolor{Blues5seq5}{RGB}{8,81,156}

\definecolor{Greens5seq1}{RGB}{237,248,233}
\definecolor{Greens5seq2}{RGB}{186,228,179}
\definecolor{Greens5seq3}{RGB}{116,196,118}
\definecolor{Greens5seq4}{RGB}{49,163,84}
\definecolor{Greens5seq5}{RGB}{0,109,44}

\definecolor{Reds5seq1}{RGB}{254,229,217}
\definecolor{Reds5seq2}{RGB}{252,174,145}
\definecolor{Reds5seq3}{RGB}{251,106,74}
\definecolor{Reds5seq4}{RGB}{222,45,38}
\definecolor{Reds5seq5}{RGB}{165,15,21}

\definecolor{Yellow}{RGB}{233, 233, 150}
\definecolor{Orange}{RGB}{255, 230, 179}
\definecolor{Blueish}{RGB}{192, 242, 217}

\definecolor{goldenyellow}{rgb}{1.0, 0.87, 0.0}

\usetikzlibrary{decorations.text,arrows.meta,bending,shapes.misc}

\tikzset{cross/.style={cross out, draw=black, minimum size=2*(#1-\pgflinewidth), inner sep=0pt, outer sep=0pt},
cross/.default={1pt}}

\pgfdeclarearrow{
  name = mytriangle,
  parameters = { \the\pgfarrowlength, \the\pgfarrowwidth, \the\pgfarrowlinewidth, \ifpgfarrowroundcap c\fi },
  defaults = { length = 0pt .85, width = 0pt 1.3, line width = 0pt .1, round = true },
  setup code = {
    \pgfarrowssettipend{1\pgfarrowlength}
    \pgfarrowssetlineend{.2\pgfarrowlength}
    \pgfarrowssetvisualbackend{0\pgfarrowlength}
    \pgfarrowssetbackend{0\pgfarrowwidth}
    \pgfarrowshullpoint{1\pgfarrowlength}{0\pgfarrowwidth}
    \pgfarrowsupperhullpoint{0\pgfarrowlength}{.5\pgfarrowwidth}
    \pgfarrowssavethe\pgfarrowwidth
    \pgfarrowssavethe\pgfarrowlength
    \pgfarrowssavethe\pgfarrowlinewidth
  },
  drawing code = {
    \ifpgfarrowroundcap\pgfsetroundjoin\fi
    \pgfsetlinewidth{\pgfarrowlinewidth}
    \pgfpathmoveto{\pgfqpoint{1\pgfarrowlength}{0\pgfarrowwidth}}
    \pgfpathlineto{\pgfpoint{0\pgfarrowlength}{.5\pgfarrowwidth}}
    \pgfpathlineto{\pgfqpoint{0\pgfarrowlength}{-.5\pgfarrowwidth}}
    \pgfpathclose
    \pgfusepathqfillstroke
  },
}

\makeatletter
\def\l@subsection#1#2{}
\def\l@subsubsection#1#2{}
\makeatother

\makeatletter
\patchcmd{\@ssect@ltx}
    {\addcontentsline{toc}{#1}{\protect\numberline{}#8}}
    {}
    {}
    {}
\makeatother

\newcommand{\SEP}{\pazocal{S}}

\newcommand{\sepp}{\mathrm{NE}}
\newcommand{\rng}{\mathrm{RNG}}

\DeclareMathAlphabet\mathbfcal{OMS}{cmsy}{b}{n}

\newcommand{\bea}{\begin{eqnarray}}
\newcommand{\eea}{\end{eqnarray}}
\newcommand{\be}{\begin{equation}}
\newcommand{\ba}{\begin{equation}\begin{aligned}}
\newcommand{\ea}{\end{aligned}\end{equation}}

\newcommand{\mX}{\mathcal{X}}

\newcommand{\ben}{\begin{enumerate}}
\newcommand{\een}{\end{enumerate}}

\newcommand{\eqdef}{\coloneqq}

\let\epsilon\varepsilon

\newcommand{\limc}[1]{{\color{red!80!black} #1}}
\let\limc\relax

\newcommand{\deff}[1]{\textbf{\emph{#1}}}

\newcommand{\Pseudo}{\pazocal{Q}}

\allowdisplaybreaks

\definecolor{quantumviolet}{HTML}{53257F}
\definecolor{quantumgray}{HTML}{555555}

\hypersetup{
colorlinks=true, linkcolor=quantumviolet, urlcolor=quantumviolet, citecolor=quantumviolet
}

 \makeatletter
\def\frontmatter@title@above{\noindent}
\def\frontmatter@title@format{\normalfont\sffamily\huge\noindent\hyphenpenalty=50000}
\def\frontmatter@preabstractspace{1.5em}

\def\frontmatter@authorformat{\def\@makefnmark{\relax}{}\vspace*{1.5em}\sffamily\raggedright\large}
\def\frontmatter@above@affilgroup{\vspace*{1.5em}}
\def\frontmatter@affiliationfont{\normalfont\sffamily\color{quantumgray}\footnotesize\noindent}
\def\frontmatter@title@produce{%
 \begingroup
  \frontmatter@title@above
  \color{quantumviolet}
  \frontmatter@title@format
  {\@title}
  \unskip
  \phantomsection\expandafter\@argswap@val\expandafter{\@title}{\addcontentsline{toc}{title}}%
  \@ifx{\@title@aux\@title@aux@cleared}{}{%
   \expandafter\frontmatter@footnote\expandafter{\@title@aux}%
  }%
  \par
  \frontmatter@title@below
 \endgroup
}%
\def\frontmatter@above@affiliation@script{%
 \addvspace{1em}%
 \skip@\@flushglue
 \@flushglue\z@ plus\hsize\relax
 \@flushglue\skip@
 \noindent
}%
   \makeatother

\begin{document}

\title{\href{https://quantum-journal.org/?s=On\%20a\%20gap\%20in\%20the\%20proof\%20of\%20the\%20generalised\%20quantum\%20Stein\%27s\%20lemma\%20and\%20its\%20consequences\%20for\%20the\%20reversibility\%20of\%20quantum\%20resources&reason=title-click}{On a gap in the proof of the generalised quantum Stein's lemma and its consequences for the reversibility of quantum resources}}

\author{Mario Berta}
\affiliation{Institute for Quantum Information, RWTH Aachen University, Aachen, Germany}
\affiliation{Department of Computing, Imperial College London, London, UK}

\author{Fernando G.~S.~L.~Brand\~{a}o}
\affiliation{Institute for Quantum Information and Matter, California Institute of Technology, Pasadena, CA, USA}
\affiliation{AWS Center for Quantum Computing, Pasadena, CA, USA}

\author{Gilad Gour}
\affiliation{Department of Mathematics and Statistics, Institute for Quantum Science and Technology, University of Calgary, AB, Canada T2N 1N4}

\author{Ludovico~Lami}
\affiliation{Institut f\"{u}r Theoretische Physik und IQST, Universit\"{a}t Ulm, Albert-Einstein-Allee 11, D-89069 Ulm, Germany}
\affiliation{QuSoft, Science Park 123, 1098 XG Amsterdam, The Netherlands}
\affiliation{Korteweg--de Vries Institute for Mathematics, University of Amsterdam, Science Park 105-107, 1098 XG Amsterdam, The Netherlands}
\affiliation{Institute for Theoretical Physics, University of Amsterdam, Science Park 904, 1098 XH Amsterdam, The Netherlands}

\author{Martin~B.~Plenio}
\affiliation{Institut f\"{u}r Theoretische Physik und IQST, Universit\"{a}t Ulm, Albert-Einstein-Allee 11, D-89069 Ulm, Germany}

\author{Bartosz Regula}
\affiliation{Mathematical Quantum Information RIKEN Hakubi Research Team, RIKEN Cluster for Pioneering Research (CPR) and RIKEN Center for Quantum Computing (RQC), Wako, Saitama 351-0198, Japan}
\affiliation{Department of Physics, Graduate School of Science, The University of Tokyo, Bunkyo-ku, Tokyo 113-0033, Japan}

\author{Marco Tomamichel}
\affiliation{Center for Quantum Technologies, National University of Singapore, Singapore}
\affiliation{Department of Electrical and Computer Engineering, College of Design and Engineering, National University of Singapore, Singapore}


\begin{abstract}
We show that the proof of the generalised quantum Stein's lemma [Brand\~ao \& Plenio, \href{https://link.springer.com/article/10.1007/s00220-010-1005-z}{Commun.\ Math.\ Phys.\ 295, 791 (2010)}] is not correct due to a gap in the argument leading to Lemma~III.9. Hence, the main achievability result of Brand\~ao \& Plenio is not known to hold. This puts into question a number of established results in the literature, in particular the reversibility of quantum entanglement [Brand\~ao \& Plenio, \href{https://link.springer.com/article/10.1007/s00220-010-1003-1}{Commun.\ Math.\ Phys.\ 295, 829 (2010)}; \href{https://www.nature.com/articles/nphys1100}{Nat.\ Phys.\ 4, 873 (2008)}] and of general quantum resources [Brand\~ao \& Gour, \href{https://journals.aps.org/prl/abstract/10.1103/PhysRevLett.115.070503}{Phys.\ Rev.\ Lett.\ 115, 070503 (2015)}] under asymptotically resource non-generating operations. We discuss potential ways to recover variants of the newly unsettled results using other approaches.
\end{abstract}

\maketitle


\section{Overview}


The question of reversibility is central in the study of how different quantum resources can be manipulated.
Asymptotic reversibility concerns the existence of processes that can interconvert two quantum states in the asymptotic limit of infinitely many independent and identically distributed (i.i.d.)\ copies. 
Whenever such reversible conversion is possible, the rate at which it can be realised constitutes a fundamental quantity that completely governs how quantum states can be converted into each other. In the theory of thermodynamics, this quantity is precisely the \emph{entropy} of a system, and generalisations of this concept underpin the study of different quantum physical phenomena.

One of the most significant contributions to the understanding of reversibility of quantum resources was the generalised quantum Stein's lemma of~\cite{Brandao2010} and the ensuing frameworks of~\cite{BrandaoPlenio1,BrandaoPlenio2,Brandao-Gour}. These results claimed that \emph{any} quantum resource satisfying a small number of physically-motivated axioms can be reversibly manipulated, as long as the constraints on resource transformations are suitably relaxed. This in particular implied that any such physical resource admits a notion of an entropy-like monotone --- a single quantity that determines asymptotic resource convertibility, and functions as the unique measure of the given resource in the asymptotic limit.

On the technical side, the main contribution of~\cite{Brandao2010} was the extension of an important result in quantum hypothesis testing --- the quantum Stein's lemma~\cite{Hiai1991,Ogawa2000}, which connects the distinguishability of quantum states with the quantum relative entropy --- to the setting of composite hypothesis testing, where the alternative hypothesis, instead of the usual i.i.d.\ form, is given by a more general convex set of quantum states. This allowed for the connection between general quantum resources and quantifiers based on the relative entropy to be established, leading directly to the aforementioned results of~\cite{BrandaoPlenio1,BrandaoPlenio2,Brandao-Gour} and revealing such entropy-based monotones as the quantities governing the reversibility of resource theories.

However, an issue has recently been found in the claimed proof of the generalised quantum Stein's lemma in~\cite{Brandao2010}. Specifically, after the appearance of the first version of the preprint~\cite{fang_preprint} that studied a related setting using the methods of~\cite{Brandao2010}, one of us identified a mistake in~\cite[Lemma~16]{fang_preprint}, which then led to the discovery that the original result~\cite[Lemma~III.9]{Brandao2010} is incorrect. This means that the main claims of~\cite{Brandao2010}, and in particular the generalised quantum Stein's lemma introduced therein, are not known to be correct, and the validity of a number of results that build on those findings is thus directly put into question.

In this work, we explain the error in~\cite[Lemma~III.9]{Brandao2010} and discuss how the main results of the reversible frameworks~\cite{BrandaoPlenio1,BrandaoPlenio2,Brandao-Gour} can be recovered in some contexts under additional assumptions. We discuss in detail where the issue occurs, what it would take to rectify it, and present some alternative strategies that could be used to potentially recover the findings of~\cite{Brandao2010}. In particular, we use other variants of a composite quantum Stein's lemma as found in~\cite{brandao_adversarial,berta_composite} to provide alternatives to the main finding of~\cite{Brandao2010} that, although more restrictive, allow us to establish a reversibility result for quantum resource theories somewhat analogous to, but weaker than, that of~\cite{BrandaoPlenio1,BrandaoPlenio2,Brandao-Gour}. We also discuss which of the results in the literature are, by virtue of relying on~\cite{Brandao2010}, no longer known to be true, and which related results are independent of \cite{Brandao2010}'s correctness.

The reader exclusively interested in the flaw of the argument leading to~\cite[Lemma~III.9]{Brandao2010} can directly jump to Section \ref{sec:gap}.


\subsection*{Summary of findings}

\begin{itemize}
    \item The proof of the generalised quantum Stein's lemma of Ref.~\cite{Brandao2010} is not correct. Due to a gap in the proof of~\cite[Lemma~III.9]{Brandao2010}, the proof of \cite[Lemma~III.7]{Brandao2010} is invalidated. Hence, the main achievability result, i.e., the direct part of \cite[Proposition~III.1]{Brandao2010}, is not known to be correct. The converse part of \cite[Proposition~III.1]{Brandao2010}, and in particular \cite[Corollary~III.3]{Brandao2010}, are unaffected. The issue is discussed in detail in Section~\ref{sec:gap}.
    
    \item This error affects other published results that directly rely on the generalised quantum Stein's lemma, and notably the framework for reversible quantum resource theories delineated in~\cite{BrandaoPlenio1,BrandaoPlenio2,Brandao-Gour}. Without any alternative 
  arguments, the main results of these works can no longer be considered proven. 
  The connection between hypothesis testing and resource reversibility is discussed in Sections~\ref{sec:asymp_manipulation}--\ref{sec:hyptest_reversibility}.
  
    \item The reversibility results are recovered whenever the generalised quantum Stein's lemma can be replaced with another Stein--type hypothesis testing result that identifies the regularised relative entropy of a resource as the optimal hypothesis testing rate. Examples of this are the restricted-measurement Stein's lemma of~\cite{brandao_adversarial}, or the variants presented in~\cite{berta_composite}. 
    The former allows us to give a general result for the reversibility of a modified resource theory of quantum entanglement that obeys certain additional restrictions on the tensor product structure in the many-copy setting. Albeit weaker than the original findings of~\cite{BrandaoPlenio1,BrandaoPlenio2,Brandao-Gour}, this provides some evidence for the validity of Stein--type results in the theory of asymptotic entanglement manipulation. 
    As a noteworthy special case, the methods of~\cite{berta_composite} allow us to give an alternative proof of the reversibility of the theory of quantum coherence, fully recovering the claims of~\cite{Brandao-Gour} in this setting~(cf.~\cite{Chitambar-reversible}). 
    This is discussed in Section~\ref{sec:alternative-stein}.
    
    \item Although the proof method of~\cite{Brandao2010} has a gap that we are unable to fix in full generality, we have also not been able to rule out the main result of~\cite[Proposition~III.1]{Brandao2010} itself. That is, there remains a possibility that the generalised quantum Stein's lemma is correct. The conclusive determination of the validity of~\cite{Brandao2010} is thus a pressing open problem.
    
    \item Some results in the literature, although using similar methods or seemingly relying on the findings of~\cite{Brandao2010}, are independent of the generalised quantum Stein's lemma and therefore unaffected by the uncovered gap in the proof. This applies, among others, to the faithfulness of the regularised relative entropy of entanglement~\cite{Piani2009}, the asymptotic equipartition property of the max-relative entropy of entanglement and other resources~\cite{Brandao2010,datta_2009-2}, and the irreversibility of entanglement under non-entangling operations~\cite{irreversibility}. 
    An important result that makes direct use of the general quantum Stein's lemma, namely the first proof of the faithfulness of squashed entanglement~\cite{faithful}, can be saved with some modifications that we describe in detail (the same conclusion has been reached independently~\cite{personal} by the authors of~\cite{faithful}). All of this is discussed in Section~\ref{sec:comments}.
\end{itemize}

We start in the following Section~\ref{sec:notation} by introducing our notation and some background material.


\section{Notation and setting}\label{sec:notation}

\subsection{Mathematics}

A quantum system $A$ is mathematically represented by a separable Hilbert space $\HH_A$. In this paper we shall always consider finite-dimensional Hilbert spaces. Quantum states of $A$ are represented by density operators, i.e.\ positive semi-definite operators on $\HH_A$ with trace one. We will denote the set of density operators on $\HH_A$ with $\D(\HH_A)$. A particular class of density operators is that formed by pure states, i.e.\ rank-one projectors $\ketbra{\psi}_A$ onto a normalised vector $\ket{\psi}_A\in \HH_A$ with $\braket{\psi|\psi}=1$. The distance between two density operators $\rho,\sigma\in \D(\HH_A)$ is quantified by the trace norm distance $ \frac12\left\|\rho-\sigma\right\|_1$, where $\|X\|_1\coloneqq \Tr \sqrt{X^\dag X}$ is the trace norm; this is endowed with operational meaning in the context of binary state discrimination by the Helstrom--Holevo theorem~\cite{HELSTROM, Holevo1976}. Quantum measurements are mathematically represented by positive operator-valued measures (POVMs) of the form $(E_x)_x$. Here $x\in \pazocal{X}$ is an index with finite range, for each $x$ the operator $E_x$ is positive semi-definite, and moreover $\sum_x E_x = \id$. More generally, physical transformations of states of system $A$ into states of system $A'$ are represented by quantum channels, i.e.\ completely positive trace preserving (CPTP) linear maps $\Lambda: \LL(\HH_A)\to \LL(\HH_{A'})$, where $\LL(\HH)$ denotes the space of linear operators (matrices) over $\HH$. The set of CPTP maps from $\LL(\HH_A)$ to $\LL(\HH_{A'})$ will be denoted by $\cptp(A\to A')$.


\subsection{Quantum entropy}\label{sec:entropy}

The (von Neumann) \deff{entropy} of a quantum state $\rho\in \D(\HH)$ is given by $S(\rho)\coloneqq -\Tr \rho \log \rho$ (with the convention that $0 \log 0 = 0$). A more general function is the (Umegaki) \deff{relative entropy}, defined for two states $\rho,\sigma\in \D(\HH)$ by
\bb
D(\rho\|\sigma) \coloneqq \left\{ \begin{array}{ll} \Tr \rho\left( \log \rho - \log \sigma\right) & \quad \supp\rho\subseteq \supp\sigma, \\[1ex] +\infty & \quad \text{otherwise.} \end{array}\right.
\label{relative_entropy}
\ee
Here, $\supp X$ denotes the support of an operator $X$. A different approach to the problem of defining a quantum relative entropy is that of measuring the two states and computing the classical Kullback--Leibler divergence~\cite{Kullback-Leibler} between the resulting probability distributions. If we denote the set of allowed POVMs by $\mathds{E}$, the quantity one obtains is the \deff{$\mathds{E}$-measured relative entropy}~\cite{Piani2009} (see also~\cite{Vedral1998})
\bb
D^{\mathds{E}} (\rho\|\sigma) \coloneqq \sup_{(E_x)_x\in \mathds{E}} \sum_x \Tr \rho E_x \log \frac{\Tr \rho E_x}{\Tr \sigma E_x}\, .
\label{measured_relative_entropy}
\ee
A special case of the above formula is obtained when $\mathds{E}=\mathds{ALL}$ comprises all POVMs, i.e.\ all (finite) collections $(E_x)_x$ of positive semi-definite operators $E_x\geq 0$ such that $\sum_x E_x=\id$~\cite{Donald1986, Petz-old, Berta2017, nonclassicality}. While $D^{\mathds{ALL}}(\rho\|\sigma) \leq D(\rho\|\sigma)$ for all $\rho,\sigma$ thanks to the data processing inequality~\cite{lieb73a,lieb73b,lieb73c,Lindblad-monotonicity}, it is known~\cite{Berta2017} that equality holds if and only if $\rho$ and $\sigma$ commute, i.e.\ if $[\rho,\sigma]=0$.

There are at least two other quantities that are related to the relative entropy and that will be widely used in this paper. One is the \deff{max-relative entropy}, given by~\cite{datta_2009}
\bb
D_{\max}(\rho\|\sigma) \coloneqq \log\min\left\{\lambda :\, \rho\leq \lambda \sigma \right\} .
\label{max-relative_entropy}
\ee
The other is the \deff{hypothesis testing relative entropy}, defined for a parameter $\epsilon\in [0,1]$ by~\cite{buscemi_2010, wang_2012}
\bb\label{eq:DH_def}
D_H^\epsilon(\rho\|\sigma) \coloneqq -\log \min\left\{ \Tr M\sigma:\, 0\leq M\leq \id,\ \Tr M\rho\geq 1-\epsilon \right\} .
\ee
The purpose of employing the two quantities is that, just as the relative entropy is often found to exactly quantify the asymptotic properties of quantum states in some operational tasks, in the one-shot regime it is $D_{\max}$ and $D^\epsilon_H$ that play a similar role in characterising the operational properties of states.
For an introduction to these quantities and their operational meaning, we refer the reader to~\cite{buscemi_2010,wang_2012,TomamichelPhD,tomamichel_2013,anshu_2019}.


\subsection{Entanglement}

Let $AB$ be a bipartite quantum system with Hilbert space $\HH_{AB}\coloneqq \HH_A\otimes \HH_B$.  The set of \deff{separable states} on $AB$ is just the convex hull of all product states, i.e.~\cite{Werner}
\begin{equation}
    \SEP_{AB} \coloneqq \mathrm{conv}\left\{ \ketbra{\psi}_A \otimes \ketbra{\phi}_B:\, \ket{\psi}_A\in \HH_A,\, \ket{\phi}_B\in \HH_B,\, \braket{\psi|\psi}=1=\braket{\phi|\phi} \right\} .
    \label{SEP}
\end{equation}
It is a fundamental fact of quantum mechanics that not all states on $AB$ are separable. A state that is not separable is called \deff{entangled}. An especially simple entangled state is the \deff{maximally entangled state} on a bipartite system $AB$ with Hilbert space $\HH_A\otimes \HH_B = \C^d\otimes \C^d$, defined by
\bb
\Phi_{AB} \coloneqq \ketbra{\Phi}_{AB}\, ,\qquad \ket{\Phi}_{AB}\coloneqq \frac{1}{\sqrt{d}}\sum_{i=1}^d \ket{ii}\, .
\label{maximally_entangled_state}
\ee
A maximally entangled state with $d=2$ is called an \deff{entanglement bit (`ebit')}. We will denote it by 
\bb
\Phi_2 \coloneqq \ketbra{\Phi_2}\, ,\qquad \ket{\Phi_2}\coloneqq \frac{1}{\sqrt2}\left( \ket{00} + \ket{11} \right) .
\ee


\subsection{General quantum resources} \label{general_resources_sec}

Entanglement is the first example of a quantum \emph{resource} to have been studied in depth. However, it is fruitful to formulate a framework capable of treating all quantum resources on an equal footing, and to establish results that can reveal broad similarities between resources of seemingly different types. The general framework of \deff{quantum resource theories}~\cite{RT-review} is designed to do so. Here, one usually identifies a family of systems of interest\,---\,with Hilbert spaces we call generically $\HH$\,---\,and over each $\HH$ a set $\MM\subseteq \D(\HH)$ of \deff{free states}, i.e.\ quantum states that are inexpensive to prepare and thus are available at will. To discuss asymptotic transformation of resources, it is important that for each $n$ we can consider the system with Hilbert space $\HH^{\otimes n}$ and the associated set of states $\MM_n\subseteq \D(\HH^{\otimes n})$. The family of sets $(\MM_n)_n$ should satisfy some elementary properties~\cite[p.~795]{Brandao2010}:
\begin{enumerate}
    \item Each $\MM_n$ is a convex and closed subset of $\D(\HH^{\otimes n})$, and hence also compact (since $\HH$ is finite dimensional).
    \item Each $\MM_n$ contains some i.i.d.\ full-rank state, i.e.\ some state of the form $\sigma^{\otimes n}$ with $\sigma>0$.
    \item The family $(\MM_{n})_n$ is closed under partial traces, i.e.\ if $\rho\in \MM_{n+1}$ then $\Tr_{n+1}\rho\in \MM_n$, where $\Tr_{n+1}$ denotes the partial trace over the last subsystem.
    \item The family $(\MM_{n})_n$ is closed under tensor products, i.e.\ if $\rho\in \MM_n$ and $\sigma\in \MM_m$ then $\rho\otimes \sigma\in \MM_{n+m}$.
    \item Each $\MM_n$ is closed under permutations, i.e.\ if $\rho\in \MM_n$ and $\pi\in S_n$ denotes an arbitrary permutation of a set of $n$ elements, then also $P_\pi\rho P_\pi^\dag\in \MM_n$, where $P_\pi$ is the unitary implementing $\pi$ over $\HH^{\otimes n}$. 
\end{enumerate}

\begin{rem}
The sets of separable states over all bipartite Hilbert spaces, defined by~\eqref{SEP}, satisfy the above Axioms~1--5.
\end{rem}

How do we quantify the amount of resource contained in a resourceful state $\rho\notin \MM$? We shall now discuss several resource monotones that can be used to this end. The \deff{generalised} (or \deff{global}) \deff{resource robustness}  
of a state $\rho\in \D(\HH)$ is defined as~\cite{VidalTarrach, harrow_2003}
\bb
R_\MM(\rho) \coloneqq \min\left\{ s\geq 0:\, \frac{1}{1+s}\left(\rho + s\sigma \right)\in \MM,\, \sigma \in \D(\HH) \right\} .
\label{g_robustness}
\ee
This is not to be confused with the \deff{standard resource robustness}, originally defined by Vidal and Tarrach for the case of entanglement as~\cite{VidalTarrach}
\bb
R_\MM^s(\rho) \coloneqq \min\left\{ s\geq 0:\, \frac{1}{1+s}\left(\rho + s\sigma \right)\in \MM,\, \sigma \in \MM \right\} .
\label{std_robustness}
\ee
Note that the only difference between~\eqref{g_robustness} and~\eqref{std_robustness} is that the state $\sigma$ in~\eqref{std_robustness} is also required to be free, while it can be an arbitrary density matrix in~\eqref{g_robustness}. Both robustnesses are faithful quantities, i.e.\ $R_\MM(\omega)=0$ if and only if $R_\MM^s(\omega)=0$, if and only if $\omega\in \MM$.

Note that it is possible to re-write~\eqref{g_robustness} as
\bb
\LR_\MM(\rho) \coloneqq \log\left(1+R_\MM(\rho)\right) = \min_{\sigma\in \MM} D_{\max}(\rho\|\sigma)\, ,
\ee
where $D_{\max}$ is the max-relative entropy given by~\eqref{max-relative_entropy}. That is, $\LR_\MM$ is just the `distance' of the state $\rho$ from the set $\MM$ as measured by the max-relative entropy. If we use instead the Umegaki relative entropy, we obtain the \deff{relative entropy of resource}, given by~\cite{Vedral1997} 
\bb
D_\MM(\rho) \coloneqq \min_{\sigma\in \MM} D(\rho\|\sigma)\, .
\label{relative_entropy_resource}
\ee
As it often happens in quantum information, the above expression needs to be regularised in order to be endowed with an operational interpretation in the asymptotic setting. That is, one needs to consider
\bb
D_\MM^\infty(\rho) \coloneqq \lim_{n\to\infty} \frac1n\, D_ \MM\left( \rho^{\otimes n} \right) .
\label{regularised_relative_entropy_resource}
\ee
Thanks to Fekete's lemma~\cite{Fekete1923}, the limit on the right-hand side exists and can be alternatively computed as an $\inf_{n\in\N}$. This is because the function $f(n) \coloneqq D_\MM(\rho^{\otimes n})$ is sub-additive, i.e.\ it satisfies $f(n+m)\leq f(n)+f(m)$, due to Axiom~4 above. Moreover, thanks to Axiom~2 we see that $D^\infty_\MM(\rho) \leq D_\MM(\rho)<\infty$ holds for every state $\rho$.


\subsection{Asymptotic manipulation of resources}\label{sec:asymp_manipulation}

To discuss how quantum resources can be transformed into each other by quantum operations, it is useful to start by looking at the important special case of entanglement. Entanglement can be manipulated with the two fundamental primitives of entanglement distillation and entanglement dilution~\cite{Bennett-distillation, Bennett-distillation-mixed, Bennett-error-correction}. The former is concerned with the transformation of a large number of i.i.d.\ copies of a certain state $\rho_{AB}$ into as many ebits as possible with vanishing error; conversely, the latter deals with the opposite process of turning a large number of ebits into as many copies of $\rho_{AB}$ as possible, again with vanishing error. In all cases, the transformation error is quantified by means of the operationally meaningful trace norm distance.

Traditionally, these transformations are effected by means of local operations assisted by classical communication (LOCCs). While the set of LOCC is well motivated operationally, in~\cite{BrandaoPlenio1, BrandaoPlenio2} a more general scheme is investigated, which involves a larger set of operations. Namely, given two bipartite systems $AB$ and $A'B'$ and some $\delta>0$, one defines the set of \deff{$\boldsymbol{\delta}$-non-entangling operations} from $AB$ to $A'B'$ as
\bb
\sepp_\delta \left(AB\to A'B'\right) \coloneqq \left\{ \Lambda\in \cptp\left(AB\to A'B'\right):\ R_\SEP\left(\Lambda(\sigma_{AB})\right) \leq \delta \ \ \forall \sigma_{AB} \in \SEP_{AB} \right\} ,
\label{delta_non-entangling}
\ee
where $R_\SEP$ is the generalised robustness of entanglement, given by~\eqref{g_robustness} with the choice $\MM=\SEP$. Since this is faithful, in the case where $\delta=0$ we obtain the set of non-entangling (or separability-preserving) operations NE~\cite{BrandaoPlenio1, BrandaoPlenio2}. 
The idea behind the definition~\eqref{delta_non-entangling} is that of 
capturing the possibility of small fluctuations in the type of physical transformations the system is undergoing. Mathematically, this translates into a framework where one allows
for the generation of some amount $\delta$ of entanglement, as quantified by the generalised robustness, so as to obtain the even larger set of operations~\eqref{delta_non-entangling}. We note that the choice of $R_\SEP$ as the quantifier of the generated entanglement here is crucial: other choices of monotones can either trivialise the whole framework, as in the case of $D_\MM(\rho)$~\cite[Sec.~V]{BrandaoPlenio2}, or make the theory asymptotically irreversible, as in the case of $R^s_\SEP$~\cite{irreversibility}; we will return to the latter issue shortly.

A necessary condition for the self-consistency of the framework is that the parameter $\delta$ should vanish in the asymptotic limit. We thus define the \deff{distillable entanglement} and the \deff{entanglement cost under asymptotically non-entangling operations} (ANE) by
\begin{align}
E_D^{\mathrm{ANE}}(\rho) &\coloneqq \sup_{(k_n)_n,\, (\delta_n)_n} \left\{ \liminf_{n\to\infty} \frac{k_n}{n}:\, \lim_{n\to\infty} \min_{\Lambda\in \sepp_{\delta_n}} \left\|\Lambda(\rho^{\otimes n}) - \Phi_2^{\otimes k_n} \right\|_1 = 0\, ,\ \lim_{n\to\infty} \delta_n = 0\right\} , \label{distillable_ANE} \\
E_C^{\mathrm{ANE}}(\rho) &\coloneqq \inf_{(k_n)_n,\, (\delta_n)_n} \left\{ \limsup_{n\to\infty} \frac{k_n}{n}:\, \lim_{n\to\infty} \min_{\Lambda\in \sepp_{\delta_n}} \left\|\Lambda(\Phi_2^{\otimes k_n}) - \rho^{\otimes n} \right\|_1 = 0\, ,\ \lim_{n\to\infty} \delta_n = 0\right\} . \label{cost_ANE}
\end{align}

In the case of an arbitrary quantum resource with sets of free states $\MM_A$ over systems $A$, the definition of \deff{$\boldsymbol{\delta}$-resource--non-generating operations} is basically analogous to~\eqref{delta_non-entangling}, with the generalised resource robustness~\eqref{g_robustness} replacing the generalised robustness of entanglement: 
\bb
\rng_\delta \left(A\to A'\right) \coloneqq \left\{ \Lambda\in \cptp\left(A\to A'\right):\ R_\MM\left(\Lambda(\sigma_{A})\right) \leq \delta \ \ \forall \sigma_{A} \in \MM_{A} \right\} .
\label{delta_rng}
\ee
Since for arbitrary resources there may or may not exist a suitable notion of a `maximally resourceful state', instead of the tasks of resource distillation and dilution it is more appropriate to define that of resource conversion. Given two systems $A,A'$ and two states $\rho\in \D(\HH_A)$ and $\omega\in \D(\HH_{A'})$, the \deff{asymptotic transformation rate $\boldsymbol{\rho\to\omega}$ under asymptotically resource non-generating operations} (ARNG) is defined by
\bb
R^{\mathrm{ARNG}}(\rho\to \omega) &\coloneqq \sup_{(k_n)_n,\, (\delta_n)_n} \left\{ \liminf_{n\to\infty} \frac{k_n}{n}:\, \lim_{n\to\infty} \min_{\Lambda\in \rng_{\delta_n}} \left\|\Lambda(\rho^{\otimes n}) - \omega^{\otimes k_n} \right\|_1 = 0\, ,\ \lim_{n\to\infty} \delta_n = 0\right\} .
\label{R_ARNG}
\ee


\subsection{Asymptotic reversibility}

As explained in the introduction, one of the key notions underlying the study of quantum resources is that of asymptotic reversibility. In the framework of~\cite{BrandaoPlenio1, BrandaoPlenio2, Brandao-Gour} and of the present paper, the set of operations considered is that of ARNG operations. In this context, we say that a quantum resource theory is \deff{asymptotically reversible} (or simply \deff{reversible}) \deff{under ARNG} if
\bb
R^{\mathrm{ARNG}}(\rho\to \omega)\, R^{\mathrm{ARNG}}(\omega \to \rho) = 1\qquad \forall\ \rho,\omega\, .
\label{reversibility_resources}
\ee
where the two states $\rho,\omega$ pertain possibly to different systems. The seminal contribution of~\cite{BrandaoPlenio1, BrandaoPlenio2, Brandao-Gour} was to establish that \emph{any} reasonable resource theory, i.e.\ one for which Axioms~1--5 are satisfied and the regularised relative entropy $D^\infty_\MM$ is non-zero, is reversible. However, the results relied crucially on the finding of~\cite{Brandao2010}, where the composite hypothesis testing of quantum states was connected to the asymptotic quantity $D^\infty_\MM$. As we discuss in more detail shortly, without the latter result, the general reversibility of quantum resources is no longer known to be true.

In the case of entanglement theory, the above relation can be rephrased as the equality between distillable entanglement and the entanglement cost under ANE. That is, Eq.~\eqref{reversibility_resources} is equivalent to the statement that
\bb
E_D^{\mathrm{ANE}}(\rho) = E_C^{\mathrm{ANE}}(\rho)\qquad \forall\ \rho\, .
\ee

Here we only concern ourselves with reversibility under \emph{asymptotically} resource non-generating operations ARNG (or ANE for the case of entanglement), as was done in the frameworks of~\cite{BrandaoPlenio1,BrandaoPlenio2,Brandao-Gour}. In some settings, such as quantum thermodynamics~\cite{Brandao-thermo,Faist2019} or coherence~\cite{Chitambar-reversible}, it is sufficient to consider strictly resource--non-generating operations, i.e.\ $\rng_\delta$ with $\delta=0$, to achieve reversibility. However, recently it was shown that such a choice leads to the \emph{irreversibility} of entanglement theory~\cite{irreversibility}, in the sense that the distillable entanglement can be strictly smaller than the entanglement cost for some states. Therefore, the use of more permissive sets of operations such as ARNG appears unavoidable if we are to establish a form of reversibility applicable to general quantum resource theories, and in particular to entanglement theory.


\subsection{Relation between hypothesis testing and reversibility}\label{sec:hyptest_reversibility}

The task of quantum hypothesis testing is concerned with distinguishing between two hypotheses given by sequences of quantum states --- the null hypothesis $(\rho_n)_n$  and alternative hypothesis $(\sigma_n)_n$, with $\rho_n, \sigma_n \in \D(\HH^{\otimes n})$ --- by performing measurements on them. Given a two-outcome measurement $\{M_n, \id - M_n\}$, one defines two types of errors:
\begin{align}
\text{type I error } &\qquad \alpha(M_n) \coloneqq \Tr \rho_n (\id - M_n)\\
\text{type II error } &\qquad \beta(M_n) \coloneqq \Tr \sigma_n M_n
\end{align}
which correspond, respectively, to the probability that we accept the alternative hypothesis when the null hypothesis is actually true, and vice versa. In the context of asymmetric hypothesis testing, with which we will be concerned here, an important problem is to understand the trade-offs between the two types of errors. In particular: how small can the type II error be when the type I error is constrained to be at most $\epsilon \in [0,1]$? 
This question gives rise precisely to the hypothesis testing relative entropy $D^\epsilon_H$ that we have already encountered in~\eqref{eq:DH_def}:
\bb
- \log \min_{\substack{0 \leq M_n \leq \id\\\alpha(M_n) \leq \epsilon }} \,\beta(M_n) = 
-\log \min\left\{ \Tr M_n\sigma_n:\, 0\leq M_n\leq \id,\ \Tr M_n\rho_n\geq 1-\epsilon \right\} 
= D_H^\epsilon(\rho_n\|\sigma_n).
\ee
One of the most important results in the characterisation of quantum hypothesis testing is quantum Stein's lemma~\cite{Hiai1991,Ogawa2000}, which tells us that, for two i.i.d.\ hypotheses $\rho_n = \rho^{\otimes n}$ and $\sigma_n = \sigma^{\otimes n}$, the asymptotic error exponent equals the relative entropy between them. Specifically, for any $\epsilon \in (0,1)$ we have that
\begin{equation}\begin{aligned}
    \lim_{n\to\infty} \frac1n D^\epsilon_H\left(\rho^{\otimes n} \| \sigma^{\otimes n}\right) = D(\rho\|\sigma).
\end{aligned}\end{equation}

A number of works have been dedicated to extending and generalising quantum Stein's lemma. Importantly, in many practically relevant contexts, the alternative hypothesis is not simply a single state $\sigma_n$ --- instead, we are tasked with determining the least error when distinguishing $\rho_n$ from a whole set of quantum states. This is commonly known as composite hypothesis testing. A seminal result in this setting was the generalised quantum Stein's lemma of~\cite{Brandao2010}, which we recall as follows.

\begin{cj}[(Generalised quantum Stein's lemma)~\cite{Brandao2010}]\label{gen_steins}
For any family of sets of quantum states $(\MM_n)_n$ satisfying Axioms~1--5 in Section~\ref{general_resources_sec}, it holds that\footnote{The work~\cite{Brandao2010} actually claims a stronger (strong converse) result, namely that~\eqref{generalised_Stein_lemma} holds not only in the limit $\ve \to 0$, but for every $\ve \in (0,1)$, and the $\liminf$ is actually a limit. Here we will not need this stronger variant.}
\begin{equation}\begin{aligned}
    \lim_{\epsilon \to 0} \liminf_{n \to \infty} \frac1n \min_{\sigma_n \in \MM_n} D^\epsilon_H \left(\rho^{\otimes n} \| \sigma_n \right) = D^\infty_{\MM}(\rho).
    \label{generalised_Stein_lemma}
\end{aligned}\end{equation}
\end{cj}

The fact that the left-hand side of~\eqref{generalised_Stein_lemma} can never exceed the right-hand side is elementary, and proved in full detail in~\cite{Brandao2010}. The converse direction is the non-trivial one.

To understand how the above result connects with asymptotic transformations of resources, it is useful to look at the case of quantum entanglement. Here, the entanglement cost and distillable entanglement under asymptotically non-entangling maps were found~\cite{BrandaoPlenio2} to correspond to two smoothed and regularised quantities: one based on the max-relative entropy, the other on the hypothesis testing relative entropy. Specifically,
\begin{equation}\begin{aligned}\label{eq:ANE_EC_ED}
    E_C^{\mathrm{ANE}}(\rho) &= \lim_{\epsilon \to 0} \limsup_{n \to \infty} \frac1n \min_{\sigma_n \in \SEP} D^\epsilon_{\max} \left(\rho^{\otimes n} \| \sigma_n\right)\\
    E_D^{\mathrm{ANE}}(\rho) &= \lim_{\epsilon \to 0} \liminf_{n \to \infty} \frac1n \min_{\sigma_n \in \SEP} D^\epsilon_{H} \left(\rho^{\otimes n} \| \sigma_n\right),
\end{aligned}\end{equation}
where $D^\epsilon_{\max} (\rho^{\otimes n} \| \sigma_n) = \min \left\{ D_{\max}(\rho'_n\|\sigma_n) \;:\; \rho'_n \in \D(\HH^{\otimes n}),\; \Delta(\rho'_n, \rho^{\otimes n})  \leq \epsilon \right\}$ and $\Delta$ is some suitable measure of distance between states, e.g.\ the trace distance $\frac12 \|\rho'_n-\rho^{\otimes n}\textbf{}\|_1$. In order to show the reversibility of entanglement, \cite{BrandaoPlenio2} then claims to prove that $E_C^{\mathrm{ANE}}(\rho) = D_\SEP^\infty(\rho) = E_D^{\mathrm{ANE}}(\rho) $. The first of these equalities holds irrespectively of the quantum Stein's lemma~\cite{BrandaoPlenio2,datta_2009-2}; however, the second one is precisely the statement of Conjecture~\ref{gen_steins}, and cannot be true without it.

More generally, the relation $E_C^{\mathrm{ANE}}(\rho) = D_\SEP^\infty(\rho) = E_D^{\mathrm{ANE}}(\rho) $ would imply that the rate of conversion between any two states $\rho, \omega \notin \SEP$ is given by
\bb
R^{\mathrm{ANE}}(\rho\to \omega) = \frac{D^\infty_\SEP(\rho)}{D^\infty_\SEP(\omega)}.
\label{rate_ratio_D_S}
\ee
Ref.~\cite{Brandao-Gour} extended this result to the case of general resource theories satisfying Axioms~1--5 of Section~\ref{general_resources_sec} --- where, as previously mentioned, there might not be a clear notion of distillable resource and resource cost. However, once again, they made an explicit use of Conjecture~\ref{gen_steins} to establish this result, and the main claim does not hold without it. We therefore see that the generalised quantum Stein's lemma underlies some of the most fundamental results in the reversibility of quantum resource theories.


\section{A gap in the proof of~\texorpdfstring{Lemma~III.9 of~\cite{Brandao2010}}{Lemma III.9 of [Brand\~ao and Plenio, CMP 2010]}}\label{sec:gap}

\subsection{Main argument}

In the proof of Lemma III.9 of~\cite{Brandao2010}, it is argued that the function (see Eq.(145) in~\cite{Brandao2010})
\bb
g(\mu_{1,n},\mu_{2,n},...)\eqdef\frac1n\left(\sum_{j}t_{j,n}(\log\mu_{j,n})^2-\left(\sum_jt_{j,n}\log\mu_{j,n}\right)^2\right) , \label{eq:wrong1}
\ee
when maximised over all distributions that satisfy $\mu_{j,n} \geq t_n := \frac\tau{1+\tau}\lambda_{\min}(\sigma)^{n-m-r}$ for all $j$ in some (finite) set $\mX^{n-m-r}$, achieves its maximum when all the $\mu_{j,n}$ except one are equal to $t$. Here $\tau > 0$ is some small constant and $n \geq m+r$, where $m$ and $r$ are natural numbers assumed to satisfy $m,r = o(n)$. The authors eventually use the above statement to conclude in Eq.~(155) in~\cite{Brandao2010} that
\begin{align}
    g(\mu_{1,n},\mu_{2,n},...)
    &\leq 1 \label{eq:wrong2}
\end{align}
for sufficiently large $n$ as long as the sequence of probability distributions $\{\mu_{j,n}\}_n$ have their minimal probability bounded from below by $t_n$.
Going beyond what is needed for their further proof, their argument in fact claims to show that $g(\mu_{1,n},\mu_{2,n},...)$ will become arbitrary small for sufficiently large $n$, i.e., $g(\mu_{1,n},\mu_{2,n},...) = o(1)$.
We argue here that the assertion in Eq.~\eqref{eq:wrong2} is incorrect and provide a counterexample. We believe that this error can be traced back to Eq.~(147) of~\cite{Brandao2010}, where a Lagrange multiplier is missing since the normalization condition $\sum_{j \in\mX^{n-m-r}} \mu_{j,n} = 1$ must be enforced.

\newcommand{\MT}[1]{\textcolor{red}{#1}}

For the purpose of illustrating the counterexample, let us simplify the notation a bit and consider the varentropy-like function
\bb
V_T(P)&\eqdef\sum_{x\in\mX} t_x(\log p_x)^2-\left(\sum_{x\in\mX} t_x\log p_x\right)^2 \;,
\ee
where $P=\{p_1,...,p_{|\mX|}\}$ is a probability vector and $T = \{t_1,...,t_{|\mX|}\}$ where $t_x \geq 0$ are arbitrary. In the special case where $T = P$ we recover the usual varentropy function 
$V(P) \eqdef\sum_{x\in\mX} p_x(\log p_x)^2-\left(\sum_{x\in\mX} p_x\log p_x\right)^2$, the variance of the log-likelihood of $P$. The claim in question can now be recast as the statement that $\frac{1}{n} V_T(P) \leq 1$ for sufficiently large $n$ and for all distributions $P$ that have their minimal eigenvalue bounded from below by $t_n$. But this is not correct, since by choosing $T = P = Q^{n-m-r}$ both i.i.d.\ according to some distribution $Q$ with minimal probability $q_{\min} = \lambda_{\min}(\sigma)$, we clearly have $V_T(P) = (n-o(n))\, V(Q)$, growing linearly with $n$ and violating the claim if furthermore $V(Q) > 1$. It is known that a discrete probability distribution $Q$ with $V(Q)>1$ exists provided that the underlying alphabet has size at least $4$~\cite[Section~2.B.1]{Reeb2015}.\footnote{To see this, use Theorem~8 of~\cite{Reeb2015} and plug $d=4$ and $r=1/8$ in Eq.~(7) of~\cite{Reeb2015}.}

This shows that the argument used to prove~\cite[Lemma III.9]{Brandao2010} is incorrect. Hence, we conclude that there is currently no proof for the middle equality in
\bb
\lim_{\epsilon \to 0} \limc{\liminf_{n \to \infty}} \frac1n \min_{\sigma_n \in \MM} D^\epsilon_{H} \left(\rho^{\otimes n} \| \sigma_n\right)\overset{?}{=}D_\MM^\infty(\rho) = \lim_{n\to\infty} \frac1n\, D_ \MM\left( \rho^{\otimes n} \right)
\ee
for general sets $\MM$, and in particular for the case of entanglement theory, where $\MM$ corresponds to the set of separable states $\SEP$.
It therefore remains an open question whether
\bb\label{eq:conjecture}
E_D^{\mathrm{ANE}}(\rho) \overset{?}{=} D_\SEP^\infty(\rho).
\ee
We remark that the equality is known to hold true for all maximally correlated states of the form $\sum_{i,j} \rho_{ij} \ketbraa{ii}{jj}$ --- indeed, even one-way LOCC operations are enough to distill from such states at a rate given by $D_\SEP^\infty$~\cite{devetak2005}. Eq.~\eqref{eq:conjecture} has recently been shown to hold true also for the antisymmetric state $\alpha_d$~\cite{dne}. At this stage we are not aware of any state that could be a potential counterexample.


\subsection{Additional comments}

One might wonder why we used an i.i.d.\ distribution to disprove \cite[Lemma III.9]{Brandao2010}, given that it is not clear that the actual probability distribution for which~\eqref{eq:wrong2} is intended to hold may be i.i.d.; however, we will now argue that the i.i.d.\ structure can indeed be seen as a special case of the claims in \cite{Brandao2010}.

For some arbitrary $\rho,\sigma\in\D(\HH)$, with $\sigma>0$, consider the i.i.d.\ sets $\MM_n \coloneqq \{\sigma^{\otimes n}\}$, which satisfy Axioms~1--5 in Section~\ref{general_resources_sec}. In the rest of this section, we refer to proofs and equations in~\cite{Brandao2010} unless otherwise specified. In the direct part of the proof of Proposition~III.1, set $y = D_{\max}(\rho\|\sigma)$. Since Eq.~(89) is satisfied (as the sequence on the l.h.s.\ is identically $0$ for all $n$), the proof should work in this case, too. One can verify that under these assumptions the state $\rho_n$ constructed in Eq.~(90) using Lemma~C.5 is simply $\rho_n = \rho^{\otimes n}$. Also, $\omega_n$ defined by Eq.~(122) is simply $\omega_n = \sigma^{\otimes (n-m-r)}$. 

The last ingredient we need is the state $\pi_n$, which Lemma~III.9 declares to be defined by Eq.~(115). To construct it, we need $\ket{\Phi_{n,m,r}}$ given in Eq.~(114) --- equivalently, in Eq.~(103).
Now, Lemma~III.7 states that $\ket{\Psi_{n,m,r}}$ can be any almost power state along $\ket{\theta}$, which is a purification of $\rho$. We are going to make the obvious choice $\ket{\Psi_{n,m,r}} = \ket{\theta}^{\otimes (n-m-r)}$. Then $\beta_0 = 1$ and $\beta_1 = \ldots = \beta_r = 0$. Thus, $\ket{\Phi_{n,m,r}}$ given in Eqs.~(103) or~(114) is actually equal to $\ket{\Psi_{n,m,r}}$, i.e.
\bb
\ket{\Phi_{n,m,r}} = \ket{\Psi_{n,m,r}} = \ket{\theta}^{\otimes(n-m-r)}\, .
\ee
Going back to Eq.~(115), one then sees that $\pi_n = \rho^{\otimes(n-m-r)}$. This means that Eq.~(137) is basically the varentropy function.

More precisely, we can now pick $\rho = \sum_x \lambda_x \ketbra{x}$ and $\sigma = \sum_x q_x \ketbra{x}$ to be diagonal. If moreover $\lambda_x = \frac{q_x^{1+s}}{\sum_{x'} q_{x'}^{1+s}}$, then Eq.~(137) yields immediately
\bb
\left|f''_n(s)\right| = \frac{n-m-r}{n} \left(\sum_x q_x (\log q_x)^2 - \left( \sum_x q_x \log q_x \right)^2 \right) = \frac{n-m-r}{n} V(Q) = V(Q) - \frac{o(n)}{n}\, .
\ee
Therefore, the above i.i.d.\ counterexample can indeed arise in the setting of the proof of Proposition~III.1 of~\cite{Brandao2010}. We thus conclude once more that the argument in Lemma III.9 of~\cite{Brandao2010} between Eq.~(145) and~(155) is erroneous.


\section{Alternative quantum Stein's lemmas}\label{sec:alternative-stein}

\subsection{Connection to resource theories}

As discussed in Section~\ref{sec:notation}, a major undesirable consequence of the gap in the proof of the generalised quantum Stein's lemma is that the reversible framework for quantum resources developed in~\cite{BrandaoPlenio1,BrandaoPlenio2,Brandao-Gour} breaks down. However, that is not to say that the methods used in these works are incorrect --- provided that a corresponding composite Stein's lemma can be established for the given resource, the reversibility results can be recovered in the same way. For clarity, let us restate the main result that remains true.

\begin{thm}[\cite{Brandao-Gour}] \label{Brandao_Gour_thm}
Consider any resource theory described by a family of sets of quantum states $(\MM_n)_n$ such that Axioms~1--5 as stated in Section~\ref{general_resources_sec} are satisfied and such that
\begin{equation}\begin{aligned}
    \lim_{\epsilon \to 0} \limc{\liminf_{n \to \infty}} \frac1n \min_{\sigma_n \in \MM_n} D^\epsilon_H \left(\rho^{\otimes n} \| \sigma_n \right) = D^\infty_{\MM}(\rho).
\end{aligned}\end{equation}
Then, for all states $\rho, \omega$ such that $D^\infty_\MM(\rho), D^\infty_\MM(\omega) > 0$, it holds that
\begin{equation}\begin{aligned}
    R^{\mathrm{ARNG}}(\rho\to \omega) = \frac{D^\infty_\MM(\rho)}{D^\infty_\MM(\omega)}.
\end{aligned}\end{equation}
\end{thm}
The problem then becomes: in what settings can a composite quantum Stein's lemma be established, allowing us to fully recover the reversibility of a given theory? There are various results and techniques in the literature that prove composite versions of Stein's lemma 
~\cite{brandao_adversarial, hayashitomamichel16, Tomamichel2018, berta_composite, Mosonyi2020}. These are not affected by the flaw in the argument from the generalised quantum Stein’s lemma in \cite{Brandao2010} and can therefore be used to recover composite Stein's lemmas for conditions related to the framework laid out in Section~\ref{general_resources_sec}. Unfortunately, in general none of these results covers the exact setting from \cite{Brandao2010} --- cf.~the discussion in \cite[Section~3.1.5]{brandao_adversarial}
--- but in the following we describe some results that are attainable.


\subsection{Coherence}

The resource theory of quantum coherence~\cite{Baumgraz2014,coherence-review} is concerned with the setting where the free states are those that are diagonal in a fixed orthonormal basis $\{\ket{i}\}$:
\begin{equation}\begin{aligned}
    \II_n \coloneqq \co \left\{ \bigotimes_{j=1}^n \ketbra{i_j} \right\}.
\end{aligned}\end{equation}
Section~3.1 in~\cite{berta_composite}, 
and in particular 
Eq.~(67) therein, shows that the asymptotic error exponent of testing the fixed state $\rho^{\otimes n}$ against the sets $\II_n$ is given by the relative entropy of coherence~\cite{Baumgraz2014}; specifically,
\begin{align}
\lim_{\epsilon \to 0} \lim_{n \to \infty} \frac1n \min_{\sigma_n \in \II_n} D^\epsilon_H \left(\rho^{\otimes n} \| \sigma_n \right) = D_{\II}(\rho),
\end{align}
where
\begin{align}
D_{\II}(\rho)=\inf_{\sigma\in \II}D(\rho\|\sigma) = D(\rho \| \Delta(\rho))
\end{align}
and $\Delta(\rho) = \sum_i \ketbra{i}\rho\ketbra{i}$. Notice that the quantity $D_\II(\rho)$ is additive, meaning that no regularisation is needed.

This fully recovers the result from~\cite{Brandao2010} for the special case of the resource theory of coherence, meaning that Theorem~\ref{Brandao_Gour_thm} can be applied and coherence theory can be shown to be reversible under the transformations ARNG. Note, however, that an even stronger result is known about this particular resource theory: quantum coherence has been shown to be reversible under the so-called dephasing-covariant operations~\cite{Chitambar-reversible}, which are a strict subset of ARNG (and even of the strictly resource non-generating operations); the proof of this fact in~\cite{Chitambar-reversible} is independent from~\cite{Brandao2010}.


\subsection{`Pseudo-entanglement'}\label{pseudo_entanglement_sec}

A type of composite quantum Stein's lemma can also be obtained for a modified version of entanglement theory. In this alternative framework, we consider bipartite systems of the form $A^nB^n$, as in standard entanglement theory, but the corresponding free states are assumed to be not only separable across the cut $A^n:B^n$, but in fact fully separable across the partition $A_1:\ldots:A_n:B_1:\ldots:B_n$. In other words, the sets of free states $\MM_n \coloneqq \Pseudo_n\subseteq \D(\HH_{AB}^{\otimes n})$ are given by
\bb
\Pseudo_n \coloneqq&\, \co\left\{ \bigotimes_{j=1}^n \sigma_{A_jB_j}^{(j)} :\, \sigma^{(j)}_{A_jB_j}\in \SEP_{A_jB_j}\ \forall\, j \right\} \\
=&\ \co\left\{ \bigotimes_{j=1}^n \psi_j^{A_j}\! \otimes \phi_j^{B_j}\!:\, \ket{\psi_j}_{A_j}\!\!\in \HH_{A_j},\, \ket{\phi_j}_{B_j}\!\!\in \HH_{B_j},\, \braket{\psi_j|\psi_j} = 1 = \braket{\phi_j|\phi_j}\ \forall\, j \right\} ,
\label{pseudo_SEP}
\ee
where $\SEP_{A_jB_j}$ is defined by~\eqref{SEP}, $\psi_j^{A_j}\coloneqq \ketbra{\psi_j}_{A_j}$, and analogously for $\phi_j^{B_j}$. Although these sets satisfy Axioms~1--5 in Section~\ref{general_resources_sec}, they are in general smaller than the sets of separable states $\SEP_{A^nB^n}$. Indeed, in~\eqref{pseudo_SEP} we require no entanglement to exist not only between Alice and Bob, but also among different Alices and different Bobs.

The reason we consider the above sets of free states is that they have the following remarkable property: for all positive integers $n,k$, 
\bb
\Tr_{A^k B^k} \left[ \id_{A^nB^n}\otimes E_{A^kB^k}\, \rho_{A^{n+k}B^{n+k}} \right]\, \in\, \R_+\!\cdot\!\Pseudo_n\quad \forall\ \rho_{A^{n+k}B^{n+k}}\in \Pseudo_{n+k}\, ,\ \ \forall\ E_{A^kB^k}\geq 0\, , 
\label{compatibility}
\ee
where $\R_+\cdot \Pseudo_n \coloneqq \left\{ \lambda \sigma:\, \lambda\geq 0,\ \sigma\in \Pseudo_n\right\}$ is the cone generated by $\Pseudo_n$. In the language of~\cite[Definition~4]{brandao_adversarial}, this means that $(\Pseudo_n)_n$ and the set of all measurements, hereafter denoted by $\mathds{M}$, form a \emph{compatible pair}. 
Therefore, we can apply immediately the theory developed in~\cite[Section~3]{brandao_adversarial}; combining it with insights from~\cite{Hayashi2002, berta_composite}, we are able to obtain the following statement.

\begin{prop} \label{k=1_prop}
For a bipartite quantum system $AB$, let $\Pseudo = (\Pseudo_n)_n$ denote the family of multi-partite states defined by~\eqref{pseudo_SEP}. Then the composite Stein's lemma
\bb
\lim_{\epsilon \to 0} \limc{\liminf_{n \to \infty}} \frac1n \min_{\sigma_n \in\Pseudo_n} D^\epsilon_H \left(\rho^{\otimes n} \| \sigma_n \right) = D_\Pseudo^\infty(\rho)
\label{generalised_Stein_lemma_pseudo_entanglement}
\ee
holds for all states $\rho=\rho_{AB}$. Therefore, by Theorem~\ref{Brandao_Gour_thm} we have that $R^{\mathrm{ARNG}}(\rho\to \omega) = \frac{D^\infty_\Pseudo(\rho)}{D^\infty_\Pseudo(\omega)}$ for all states $\rho, \omega$ such that $D^\infty_\Pseudo(\rho), D^\infty_\Pseudo(\omega) > 0$.
\end{prop}

\begin{proof}
By the discussion following the statement of Conjecture~\ref{gen_steins}, it suffices to prove that the left-hand side of~\eqref{generalised_Stein_lemma_pseudo_entanglement} is never smaller than the right-hand side. To this end, we start by applying~\cite[Theorem~16 and Lemma~13]{brandao_adversarial}, which yield the identity
\bb
\lim_{\epsilon \to 0} \limc{\liminf_{n \to \infty}} \frac1n \min_{\sigma_n \in \Pseudo_n} D^\epsilon_H \left(\rho^{\otimes n} \| \sigma_n \right) &=\lim_{\epsilon \to 0} \limc{\limsup_{n \to \infty}} \frac1n \min_{\sigma_n \in \Pseudo_n} D^\epsilon_H \left(\rho^{\otimes n} \| \sigma_n \right) \\&=
\lim_{n\to\infty} \frac1n \min_{\sigma_n\in \Pseudo_n} D^{\mathds{ALL}}( \rho^{\otimes n} \| \sigma_n)\, ,
\ee
which holds provided that the limit on the right-hand side exists. Here, $D^{\mathds{ALL}}$ is the measured relative entropy~\eqref{measured_relative_entropy}. The proof is complete if we argue that 
\bb
\lim_{n\to\infty} \frac1n \min_{\sigma_n\in \Pseudo_n} D^{\mathds{ALL}}( \rho^{\otimes n} \| \sigma_n) = D_\Pseudo^\infty(\rho)\, ,
\label{asymptotically_measuring_immaterial}
\ee
where the right-hand side is defined by~\eqref{regularised_relative_entropy_resource}. What~\eqref{asymptotically_measuring_immaterial} is telling us is that the fact that we are forced to carry out a measurement before computing the relative entropy of resource is asymptotically immaterial.

To prove~\eqref{asymptotically_measuring_immaterial} we make use of Axioms~1 and~5 in Section~\ref{general_resources_sec}, which, as mentioned, are satisfied for the sets $\Pseudo_n$ defined by~\eqref{pseudo_SEP}. For a start, due to the convexity and permutation invariance of $\Pseudo_n$ and thanks to the convexity of the measured relative entropy, we can take $\sigma_n$ in~\eqref{asymptotically_measuring_immaterial} to be permutation invariant, i.e.\ such that $\sigma_n = P_\pi \sigma_n P_\pi^\dag$ for all permutations $\pi\in S_n$, where $P_\pi$ in the unitary implementing $\pi$ over $\HH_{AB}^{\otimes n}$. In this case, it is argued in~\cite[Eq.~(48)]{berta_composite} that there exists a universal polynomial $q(n)$ such that 
\bb
|\spec(\sigma_n)| \leq q(n)\, ,
\ee
where $|\spec(\sigma_n)|$ is the size of the spectrum of $\sigma_n$, i.e.\ the number of different eigenvalues it has. Hayashi's pinching inequality~\cite{Hayashi2002} implies, e.g.\ via~\cite[Eq.~(47)]{berta_composite}, that
\bb
D(\rho^{\otimes n} \| \sigma_n) - \log q(n) \leq D\left( \pazocal{P}_{\sigma_n}(\rho^{\otimes n})\, \big\|\, \sigma_n\right) \leq D(\rho^{\otimes n} \| \sigma_n)\, ,
\label{pinching_inequality}
\ee
where $\pazocal{P}_{\sigma_n} (X) \coloneqq \lim_{T\to\infty} \frac{1}{2T} \int_{-T}^{T} dt\, \sigma_n^{it} X \sigma_n^{-it}$ is the pinching operator associated with $\sigma_n$. We now write that
\bb
\min_{\sigma_n\in \Pseudo_n} \frac1n\, D^{\mathds{ALL}}( \rho^{\otimes n} \| \sigma_n) &= \min_{\substack{\sigma_n\in \Pseudo_n,\\[1pt] \sigma_n = P_\pi \sigma_n P_\pi\, \forall\pi\in S_n}} \frac1n\, D^{\mathds{ALL}}( \rho^{\otimes n} \| \sigma_n)\\
&\textgeq{(i)} \min_{\substack{\sigma_n\in \Pseudo_n,\\[1pt] \sigma_n = P_\pi \sigma_n P_\pi\, \forall\pi\in S_n}} \frac1n\, D^{\mathds{ALL}}\!\left( \pazocal{P}_{\sigma_n}(\rho^{\otimes n})\, \big\|\, \sigma_n\right) \\
&\texteq{(ii)} \min_{\substack{\sigma_n\in \Pseudo_n,\\[1pt] \sigma_n = P_\pi \sigma_n P_\pi\, \forall\pi\in S_n}} \frac1n\, D\!\left( \pazocal{P}_{\sigma_n}(\rho^{\otimes n})\, \big\|\, \sigma_n\right) \\
&\textgeq{(iii)} \min_{\substack{\sigma_n\in \Pseudo_n,\\[1pt] \sigma_n = P_\pi \sigma_n P_\pi\, \forall\pi\in S_n}} \left( \frac1n\, D(\rho^{\otimes n} \| \sigma_n) - \frac{\log q(n)}{n} \right) \\
&\texteq{(iv)} \min_{\sigma_n\in \Pseudo_n} \frac1n\, D(\rho^{\otimes n} \| \sigma_n) - \frac{\log q(n)}{n} \\
&= \frac1n\, D_\Pseudo(\rho^{\otimes n}) - \frac{\log q(n)}{n}
\ee
Here: (i)~is a consequence of the data processing inequality for the measured relative entropy; (ii)~follows by observing that since $\left[ \pazocal{P}_{\sigma_n}(\rho^{\otimes n}),\, \sigma_n\right]=0$ we have that $D^{\mathds{ALL}}\!\left( \pazocal{P}_{\sigma_n}(\rho^{\otimes n})\, \big\|\, \sigma_n\right) = D\!\left( \pazocal{P}_{\sigma_n}(\rho^{\otimes n})\, \big\|\, \sigma_n\right)$; (iii)~is just~\eqref{pinching_inequality}; and finally (iv)~is because, once again due to the convexity and permutation invariance of $\Pseudo_n$ and to the convexity of the relative entropy, we can take $\sigma_n$ to be permutation invariant when minimising $D(\rho^{\otimes n}\|\sigma_n)$ over $\sigma_n\in \Pseudo_n$.

Taking the limit $n\to\infty$ of the above inequality and remembering that $q$ is a polynomial yields
\bb
\liminf_{n\to\infty} \frac1n \min_{\sigma_n\in \Pseudo_n} D^{\mathds{ALL}}( \rho^{\otimes n} \| \sigma_n) \geq D_\Pseudo^\infty(\rho)\, .
\ee
Since the converse inequality $\limsup_{n\to\infty} \frac1n \min_{\sigma_n\in \Pseudo_n} D^{\mathds{ALL}}( \rho^{\otimes n} \| \sigma_n) \leq D_\Pseudo^\infty(\rho)$ trivially holds, this completes the proof of~\eqref{asymptotically_measuring_immaterial} and of the proposition.
\end{proof}

We state~\eqref{asymptotically_measuring_immaterial} as a separate corollary, due to its independent interest.

\begin{cor}
Let $\MM=(\MM_n)_n$ be a family of sets that satisfies Axioms~1 and~5 as stated in Section~\ref{general_resources_sec}. Then
\bb
\lim_{n\to\infty} \frac1n \min_{\sigma_n\in \MM_n} D^{\mathds{ALL}}( \rho^{\otimes n} \| \sigma_n) = D_\MM^\infty(\rho)
\ee
holds for all states $\rho$. Here, $D^{\mathds{ALL}}$ is the measured relative entropy~\eqref{measured_relative_entropy}, and $D_\MM^\infty$ is defined by~\eqref{regularised_relative_entropy_resource}.
\end{cor}

\subsection{Pseudo-entanglement in blocks}\label{sec:pseudo_blocks}

It is possible to generalise Proposition~\ref{k=1_prop} by modifying the set of free states~\eqref{pseudo_SEP} so as to include states that are still separable between Alice and Bob, but possibly entangled among blocks of $k$ Alices and $k$ Bobs. By doing so, it can be shown that one can recover, in the limit $k\to\infty$, a formal identity somewhat analogous to~\eqref{generalised_Stein_lemma_pseudo_entanglement} but featuring 
the actual regularised relative entropy of entanglement $D^\infty_\SEP$ instead of the less transparent quantity encountered on the right-hand side of~\eqref{generalised_Stein_lemma_pseudo_entanglement}. This is promising as it explicitly shows that $D^\infty_\SEP$ can be obtained as a limit of the pseudo-entanglement approach of Section~\ref{pseudo_entanglement_sec}; however, the end result is not a proper Stein's lemma, and thus does not represent a solution of Conjecture~\ref{gen_steins} in the case of entanglement theory. Whether the main result of~\cite{Brandao2010} can be recovered by following this strategy is not clear to us, but the approach might offer some insights into the problem.

\begin{prop} \label{arbitrary_k_prop}
For a bipartite quantum system $AB$, let $\Pseudo^{(k)} = \big(\Pseudo_n^{(k)}\big)_n$ denote the family of multi-partite states defined by~\eqref{pseudo_SEP} upon making the substitutions $A\mapsto A^{(k)}$ and $B\mapsto B^{(k)}$, where $A^{(k)}=A_1\ldots A_k$ and $B^{(k)} = B_1\ldots B_k$ are made of $k$ copies of $A$ and $B$, respectively. Then for an arbitrary state $\rho=\rho_{AB}$ it holds that
\bb
\lim_{k\to\infty} \lim_{\epsilon \to 0} \limc{\liminf_{n \to \infty}} \frac{1}{nk} \min_{\sigma \in \Pseudo_n^{(k)}} D^\epsilon_H \left(\rho^{\otimes n k}\, \Big\| \,\sigma \right) = D_\SEP^\infty(\rho)\, .
\label{arbitrary_k_Stein_lemma}
\ee
\end{prop}

\begin{proof}
Proposition~\ref{k=1_prop} applied with the substitution $\rho\mapsto \rho^{\otimes k}$ yields immediately
\bb
\lim_{\epsilon \to 0} \limc{\liminf_{n \to \infty}} \frac{1}{n} \min_{\sigma \in \Pseudo_n^{(k)}} D^\epsilon_H \left(\rho^{\otimes n k}\, \Big\| \,\sigma \right) = D_{\Pseudo^{(k)}}^\infty\big(\rho^{\otimes k}\big)\, .
\ee
Let us divide both sides by $k$ and take the limit $k\to\infty$. On the one hand,
\bb
\limsup_{k\to\infty} \frac1k \lim_{\epsilon \to 0} \limc{\liminf_{n \to \infty}} \frac{1}{n} \min_{\sigma \in \Pseudo_n^{(k)}} D^\epsilon_H \left(\rho^{\otimes n k}\, \Big\| \,\sigma \right) &= \limsup_{k\to\infty} \frac1k\, D_{\Pseudo^{(k)}}^\infty\big(\rho^{\otimes k}\big) \\
&\textleq{(i)} \limsup_{k\to\infty} \frac1k\, D_{\Pseudo^{(k)}_1}\big(\rho^{\otimes k}\big) \\
&\texteq{(ii)} \limsup_{k\to\infty} \frac1k\, D_{\SEP}\big(\rho^{\otimes k}\big) \\
&= D_{\SEP}^\infty (\rho)\, ,
\label{arbitrary_k_limsup}
\ee
where (i)~is because of Fekete's lemma~\cite{Fekete1923} (see the discussion after~\eqref{regularised_relative_entropy_resource}) and (ii)~holds due to the fact that at the first level $\Pseudo^{(k)}_1 = \SEP$ comprises all separable states of the system $A^{(k)}:B^{(k)} = A_1\ldots A_k:B_1\ldots B_k$. On the other hand,
\bb
\liminf_{k\to\infty} \frac1k \lim_{\epsilon \to 0} \limc{\liminf_{n \to \infty}} \frac{1}{n} \min_{\sigma \in \Pseudo_n^{(k)}} D^\epsilon_H \left(\rho^{\otimes n k}\, \Big\| \,\sigma \right) &= \liminf_{k\to\infty} \frac1k\, D_{\Pseudo^{(k)}}^\infty\big(\rho^{\otimes k}\big) \\
&\textgeq{(iii)} \liminf_{N\to\infty} \frac1N\, D_{\SEP}\big(\rho^{\otimes N}\big) \\
&= D_{\SEP}^\infty (\rho)\, ,
\label{arbitrary_k_liminf}
\ee
where (iii)~is a consequence of the fact that for all $n$ the set $\Pseudo^{(k)}_n$ contains only separable states of $N=nk$ Alices vs.\ $N=nk$ Bobs, in formula
\bb
\Pseudo^{(k)}_n \subset \SEP_{(A^{(k)})_1\ldots (A^{(k)})_n : (B^{(k)})_1\ldots (B^{(k)})_n} = \SEP_{A_1\ldots A_{nk}:B_1\ldots B_{nk}}\, .
\ee
Putting~\eqref{arbitrary_k_limsup} and~\eqref{arbitrary_k_liminf} together proves~\eqref{arbitrary_k_Stein_lemma}.
\end{proof}


\subsection{Entanglement distillability\texorpdfstring{\,---\,}{---}via pseudo-Stein's lemma}\label{sec:entanglement-pseudo-stein}

Although it may not be possible to recover Brand\~{a}o and Plenio's original Stein's lemma and the associated reversibility of the asymptotic theory of entanglement manipulation using results obtained afterwards with different methods, we can do something a little less ambitious. We will now see that it is possible to use the results from~\cite{brandao_adversarial} to state an achievability result for entanglement distillation under asymptotically non-entangling operations.

The fundamental condition underpinning the composite quantum Stein's lemma of~\cite{brandao_adversarial} is that of \emph{compatibility}, expressed by~\eqref{compatibility}. In Section~\ref{pseudo_entanglement_sec}, we could satisfy it because we considered the set of all measurements paired with the restricted set of fully separable states. In this section, we will follow the somewhat opposite strategy: we want to take $\MM$ to be the standard set of separable states, as defined by~\eqref{SEP}; to do so, we will need to constrain the set of available measurements. Namely, given a bipartite system $AB$ consider the set of \deff{separable measurements}
\bb
\mathds{SEP}_{AB} \coloneqq \left\{ (E_x)_{x\in\pazocal{X}}:\, |\pazocal{X}|<\infty,\ E_x\in \R_+\!\cdot\! \SEP_{AB}\ \ \forall\, x\in\pazocal{X},\ \sum_x E_x=\id \right\} .
\label{SEP_measurements}
\ee
When there is no ambiguity regarding the underlying systems, we will omit the subscripts and write simply $\mathds{SEP}$. According to~\eqref{measured_relative_entropy}, the associated measured relative entropy is given by
\bb
D^{\mathds{SEP}} (\rho\|\sigma) \coloneqq \sup_{(E_x)_x\in \mathds{SEP}} \sum_x \Tr \rho E_x \log \frac{\Tr \rho E_x}{\Tr \sigma E_x}\, .
\ee
We can use this quantity to define the \deff{separably measured relative entropy of entanglement} and its regularisation by setting
\bb
D_\SEP^{\mathds{SEP}}(\rho) &\coloneqq \min_{\sigma\in \SEP_{AB}} D^{\mathds{SEP}}(\rho\|\sigma)\, ,\\
D_\SEP^{\mathds{SEP},\infty}(\rho) &\coloneqq \lim_{n\to\infty} \frac1n D^{\mathds{SEP}}_\SEP\left(\rho^{\otimes n}\right) .
\label{separably_measured_REE}
\ee
In a seminal work, Piani~\cite{Piani2009} introduced the above quantities and proved that the limit in the second line of~\eqref{separably_measured_REE} exists and is equal to $\sup_{n\in \N} \frac1n D^{\mathds{SEP}}_\SEP\left(\rho^{\otimes n}\right)$. In the same paper~\cite{Piani2009} it is also argued that
\bb
D_\SEP^{\mathds{SEP},\infty}(\rho) \geq D_\SEP^{\mathds{SEP}}(\rho) > 0
\label{Piani_inequality}
\ee
holds for all entangled states $\rho\notin \SEP$. In other words, the separably measured relative entropy of entanglement is faithful.

As mentioned, the key property that we recover by sacrificing some measurements is compatibility with the full set of separable states. Namely, we obtain that for all positive integers $n,k$ (cf.~\eqref{compatibility})
\bb
\Tr_{A^k B^k} \left[ \id_{A^nB^n}\otimes E_{A^kB^k}\, \rho_{A^{n+k}B^{n+k}} \right]\, \in\, \R_+\!\cdot\!\SEP_{A^nB^n}\quad \forall\ \rho_{A^{n+k}B^{n+k}}\in \SEP_{A^{n+k}B^{n+k}} ,\ \forall\ E_{A^kB^k}\in (E_x)_x\in \mathds{SEP}_{A^kB^k} .
\label{compatibility_SEP}
\ee
We can now state the pseudo-Stein's lemma for entanglement theory promised at the beginning of this section.

\begin{prop} \label{pseudo_Stein_entanglement_prop}
For all states $\rho$, the distillable entanglement under ANE defined by~\eqref{distillable_ANE} satisfies that
\bb
E_D^{\mathrm{ANE}}(\rho) = \lim_{\epsilon \to 0} \limc{\liminf_{n \to \infty}} \frac1n \min_{\sigma_n \in \SEP} D^\epsilon_{H} \left(\rho^{\otimes n} \| \sigma_n\right) \geq D_\SEP^{\mathds{SEP},\infty}(\rho) \geq D_\SEP^{\mathds{SEP}}(\rho)\, ,
\ee
where the quantities $D_\SEP^{\mathds{SEP},\infty}$ and $D_\SEP^{\mathds{SEP}}$ are defined by~\eqref{separably_measured_REE}. In particular, we have
\bb
E_D^{\mathrm{ANE}}(\rho) > 0 \quad \forall\ \rho\notin \SEP\, .
\label{no_bound_entanglement_ANE}
\ee
In other words, in the asymptotic theory of entanglement manipulation under ANE every entangled state is distillable. The same is true if one replaces ANE with the strictly non-entangling operations NE in the above.
\end{prop}

\begin{proof}
Thanks to~\eqref{eq:ANE_EC_ED} and~\eqref{Piani_inequality}, we only need to prove that $\lim_{\epsilon \to 0} \lim_{n \to \infty} \frac1n \min_{\sigma_n \in \SEP} D^\epsilon_{H} \left(\rho^{\otimes n} \| \sigma_n\right) \geq D_\SEP^{\mathds{SEP},\infty}(\rho)$. Due to the compatibility between separable states and separable measurements established in~\eqref{compatibility_SEP}, this is simply a consequence of~\cite[Theorem~16]{brandao_adversarial} (see also~\cite[Lemma~13]{brandao_adversarial}).

The result applies in the same way also to strictly non-entangling maps NE because the rate of distillation is actually the same under the two sets; that is,
\begin{equation}\begin{aligned}\label{eq:distillable_ane}
    E_D^{\mathrm{ANE}}(\rho) = E_D^{\mathrm{NE}}(\rho)
\end{aligned}\end{equation}
for any state $\rho$. A proof of this fact that avoids using the results of~\cite{Brandao2010} can be found in~\cite[Lemma~S17]{irreversibility}.
\end{proof}

The inequality in~\eqref{no_bound_entanglement_ANE} is remarkable because it demonstrates that the theory of entanglement manipulation under NE/ANE is fundamentally different from that under LOCCs. Indeed, it is well known that the latter admits \emph{bound entanglement}, i.e.\ that there exist entangled states that are undistillable under LOCCs~\cite{HorodeckiBound}. This is the case, for example, for all entangled states whose partial transposition~\cite{PeresPPT} is positive semi-definite --- such states are usually called `PPT states'.\footnote{It is worth mentioning in passing that an outstanding open problem of quantum information theory is to decide whether all LOCC-undistillable states are PPT~\cite{Horodecki-open-problems}.} However, Proposition~\ref{pseudo_Stein_entanglement_prop}, and in particular~\eqref{no_bound_entanglement_ANE}, imply that no bound entanglement can exist under NE or ANE operations: while distillation may not be possible at a rate equal to the entanglement cost of the state --- this was precisely the result of Brand\~{a}o and Plenio~\cite{BrandaoPlenio1, BrandaoPlenio2} --- it is always possible at \emph{some} non-zero rate provided that the state is entangled at all.

Although the bound from Proposition~\ref{pseudo_Stein_entanglement_prop} is of fundamental importance, it might be difficult to evaluate in practice. Let us remark a simpler single-letter bound that also applies to this setting: the quantity $E_\eta$, introduced in~\cite{irreversibility-PPT} as a lower bound on the regularised relative entropy of entanglement that is efficiently computable as a semi-definite program (SDP), can also serve as an achievability bound for entanglement distillation regardless of~\cite{Brandao2010}'s generalised Stein's lemma. Specifically,
\begin{equation}\begin{aligned}
        E_D^{\mathrm{ANE}}(\rho) &= \lim_{\epsilon \to 0} \limc{\liminf_{n \to \infty}} \frac1n \min_{\sigma_n \in \SEP} D^\epsilon_{H} \left(\rho^{\otimes n} \| \sigma_n\right)\\
        &\geq \liminf_{n \to \infty} \frac1n \min_{\sigma_n \in \SEP} D^0_{H} \left(\rho^{\otimes n} \| \sigma_n\right)\\
        &\textgeq{(i)} \liminf_{n \to \infty} \frac1n E_\eta(\rho^{\otimes n})\\
        &\texteq{(ii)}  E_\eta(\rho),
\end{aligned}\end{equation}
where the inequality~(i) is shown in~\cite[Proposition~1]{irreversibility-PPT}, and the equality in~(ii) follows from the additivity of $E_\eta$~\cite{irreversibility-PPT}. This bound can be significantly easier to evaluate than ones based on the measured relative entropy, but it is often weaker --- for instance, it trivialises for any full-rank state.

We note that the results of this section apply also to the case where separable states and separable measurements are replaced with the sets of PPT states and PPT measurements, giving rise to the class of (asymptotically) PPT-preserving operations rather than the (asymptotically) non-entangling maps that we used in our discussion.


\section{Comments on other results}\label{sec:comments}

\subsection{Overview}

Beyond the problem of resource reversibility discussed above, a number of results in the literature are affected by the fact that the generalised quantum Stein's lemma of~\cite{Brandao2010} (Conjecture~\ref{gen_steins}) can no longer be verified to be correct. At the same time, some related results are actually independent of Conjecture~\ref{gen_steins} or they admit known alternative proofs that do not rely on the problematic lemma. In this section we aim to clarify the validity of several other results.


\subsection{Relative entropy and rates of resource transformations}

Although the lack of a universally applicable generalised quantum Stein's lemma prevents us from establishing an exact expression for the rate of asymptotic resource transformations $R^{\mathrm{ARNG}}$ valid in general resource theories, it is worth noting that one of the two inequalities in~\eqref{rate_ratio_D_S} is nevertheless still true: specifically, the ratio of regularised relative entropies $D^\infty_\MM$ always serves as an upper bound for conversion rates. This is well known~\cite{Horodecki2001,horodecki_2012}, and follows essentially from the asymptotic continuity of $D_\MM$~\cite{Synak2006,tightuniform}; we provide a self-contained proof below for the sake of completeness.

\begin{lemma}
For all states $\rho,\omega$ and for all sequences of sets of states $(\MM_n)_n$ satisfying Axioms~1--5 in Section~\ref{general_resources_sec}, it holds that
\bb
R^{\mathrm{ARNG}}(\rho\to \omega) \leq \frac{D^\infty_\MM(\rho)}{D^\infty_\MM(\omega)}
\ee
provided that $D^\infty_\MM(\omega)>0$.
\end{lemma}

\begin{proof}
The argument is obtained from the one on pp.~841--842 of~\cite{BrandaoPlenio2} with minor modifications. For an arbitrary sequence $\Lambda_n\in \rng_{\delta_n}$ of protocols appearing in~\eqref{R_ARNG}, denoting $\epsilon_n\coloneqq \frac12 \left\|\Lambda_n(\rho^{\otimes n}) - \omega^{\otimes k_n} \right\|_1$ we can write that
\bb
D_\MM\left(\rho^{\otimes n}\right) &\textgeq{(i)} D_\MM\left(\Lambda_n(\rho^{\otimes n})\right) + \log(1-\delta_n) \\
&\textgeq{(ii)} D_\MM\big( \omega^{\otimes k_n}\big) 
+ \epsilon_n k_n \log\lambda - g(\epsilon_n) + \log(1-\delta_n) \, .
\ee
Here, (i)~follows from the fact that $D_\MM(\Lambda(\rho))\leq D_\MM(\rho) - \log(1-\delta)$ for all $\Lambda\in \rng_\delta$, as can be seen by noticing that Lemma~IV.5 in~\cite{BrandaoPlenio2} holds not only for entanglement but also for the slightly more general case at hand here; in~(ii) we employed the asymptotic continuity of $D_\MM$~\cite[Lemma~7]{tightuniform}, 
introducing the auxiliary function $g(x)\coloneqq (1+x)\log (1+x) - x\log x$  and using the fact that the existence of a full-rank state $\sigma$ such that $\sigma^{\otimes m} \in \MM_m \; \forall m$ --- ensured by Axiom~2 of Section~\ref{general_resources_sec} --- guarantees that the relative entropy of resource of any state $\tau$ acting on the same space as $\omega^{\otimes k_n}$ is bounded as
$D_{\MM_{k_n}}(\tau) \leq D(\tau \| \sigma^{\otimes k_n}) \leq - \log \lambda^{k_n}$,  with $\lambda > 0$ standing for the smallest eigenvalue of $\sigma$. 
Dividing both sides by $n$ and taking the limit yields
\bb
D_\MM^\infty(\rho) &= \lim_{n\to\infty} \frac1n D_\MM\left(\rho^{\otimes n}\right) \\
&\geq \liminf_{n\to\infty} \frac{k_n}{n} \left(\frac{1}{k_n} D_\MM\big(\omega^{\otimes k_n}\big)  
+ \epsilon_n \log \lambda\right) - \limsup_{n\to\infty}\left( \frac1n g(\epsilon_n) - \frac1n \log(1-\delta_n)\right) \\
&\texteq{(iii)} \left( \liminf_{n\to\infty} \frac{k_n}{n}\right) D_\MM^\infty(\omega)\, ,
\ee
where in~(iii) we assumed that $\lim_{n\to\infty} \delta_n = 0$ and also $\lim_{n\to\infty} \epsilon_n=0$. Since the sequence of operations $\Lambda_n\in \rng_{\delta_n}$ is otherwise arbitrary, the claim follows.
\end{proof}


\subsection{Asymptotic properties of quantum entanglement}

One of the 
applications of the methods used to study the generalised quantum Stein's lemma, originally discussed in~\cite{Brandao2010}, is the fact that the regularised relative entropy of entanglement is faithful; that is, for any state $\rho \notin \SEP$, it holds that
\begin{equation}\begin{aligned}\label{eq:relent_faithful}
D^\infty_\SEP(\rho) > 0.     
\end{aligned}\end{equation} 
This result is a priori not obvious, as it is well known that the relative entropy of entanglement can be strictly subadditive~\cite{Werner-symmetry}, and furthermore there exist choices of sets $(\MM_n)_n$ such that $D^\infty_\MM(\rho) = 0$ even for resourceful states~\cite{gour_2009}. The investigation of this result in~\cite{Brandao2010} only uses the \emph{converse} part of the generalised quantum Stein's lemma, which is unaffected by the gap in the proof of Conjecture~\ref{gen_steins}.

Furthermore, as can already be deduced from our discussion in Section~\ref{sec:entanglement-pseudo-stein}, the work of Piani~\cite{Piani2009} provides an alternative proof of the faithfulness of $D^\infty_\SEP$. Indeed, this result applies not only to the theory of entanglement, but also to more general resource theories whose free measurements contain informationally complete POVMs (see~\cite[Theorem~1 and the subsequent discussion]{Piani2009}).

Another result concerning the asymptotic behaviour of entangled states is the fact that the distance of any entangled state from the set of separable states necessarily goes to $1$ in the many-copy limit; specifically, for any $\rho \notin \SEP$,
\begin{equation}\begin{aligned}
    \lim_{n\to\infty} \min_{\sigma_n \in \SEP} \frac12 \left\| \rho^{\otimes n} - \sigma_n \right\|_1 = 1.
    \label{Beigi_Shor}
\end{aligned}\end{equation}
The above identity~\eqref{Beigi_Shor} shows, for instance, that the set of states with positive partial transpose cannot provide a good approximation to the set of separable states asymptotically.
As discussed in the original work of Beigi and Shor~\cite{beigi_2010}, this result can be considered to be a consequence of the fact that all entangled states are distillable under non-entangling operations, which is in turn a consequence of the generalised quantum Stein's lemma.
However, we have already seen in our Proposition~\ref{pseudo_Stein_entanglement_prop} that this distillability can be recovered using a different approach. In fact, \cite{beigi_2010} already gave an alternative proof of~\eqref{Beigi_Shor} that does not use the generalised quantum Stein's lemma or the distillability properties of entanglement, so this finding is unaffected.

There are also results in the literature that ostensibly rely on the findings of~\cite{Brandao2010}, but they actually only use the fact that
\begin{equation}\begin{aligned}\label{eq:dmax_regularised}
    \lim_{\epsilon \to 0}  \limc{\limsup_{n\to\infty}} \frac1n \min_{\sigma_n \in \SEP} D^\epsilon_{\max} \left(\rho^{\otimes n} \| \sigma_n\right) = D_\SEP^\infty(\rho).
\end{aligned}\end{equation}
This is the case, for instance, for the result on disentanglement cost of quantum states~\cite{berta_2018}. We reiterate that the relation \eqref{eq:dmax_regularised} is still true; the proof of this fact in~\cite{BrandaoPlenio2} can be verified not to require the generalised quantum Stein's lemma to hold, and an alternative proof can be found in~\cite{datta_2009-2}. This applies also to more general resource theories under Axioms~1--5 as stated in Section~\ref{general_resources_sec}.

Finally, a recent result that complements the reversibility of entanglement studied in~\cite{BrandaoPlenio1,BrandaoPlenio2} is the \emph{ir}reversibility of entanglement shown in~\cite{irreversibility} under non-entangling operations and variants thereof. This result does not make use of the generalised quantum Stein's lemma and holds regardless of the validity of~\cite{Brandao2010}. We also stress that the work~\cite{irreversibility} does not imply that entanglement cannot be asymptotically reversible in the sense of~\cite{BrandaoPlenio1,BrandaoPlenio2}: entanglement reversibility under strictly non-entangling operations (which was left as an open question in~\cite{BrandaoPlenio2}) is now ruled out, but there still remains a possibility that the larger class of operations used in~\cite{BrandaoPlenio1,BrandaoPlenio2} is capable of transforming entangled states reversibly. Note, however, that~\cite{irreversibility} shows that any reversible entanglement framework must generate exponentially large amounts of entanglement according to some entanglement monotones such as the negativity, meaning that the conjectured results of~\cite{BrandaoPlenio1,BrandaoPlenio2} heavily depend on the specific choice of the generalised robustness $R_\SEP$ as the quantifier of the generated entanglement.


\subsection{Faithfulness of squashed entanglement}

A result related to the issues discussed above is the faithfulness of another measure of entanglement, the so-called squashed entanglement $E_{\rm sq}$. This quantity is defined as~\cite{squashed,Tucci1999}
\begin{equation}\begin{aligned}
    E_{\rm sq}(\rho_{AB}) \coloneqq \inf_{\rho_{ABE}} \, \frac12 I(A:B | E)_\rho,
\end{aligned}\end{equation}
where the minimisation is over all quantum states $\rho_{ABE}$ such that $\Tr_E \rho_{ABE} = \rho_{AB}$, and where $I(A:B | E) = D(\rho_{ABE} \| \rho_{A} \otimes \rho_{BE}) - D(\rho_{AE} \| \rho_{A} \otimes \rho_{E})$ is the conditional mutual information.

The first proof of the faithfulness of $E_{\rm sq}$, appearing in~\cite{faithful}, made explicit use of the generalised quantum Stein's lemma (Conjecture~\ref{gen_steins}), and therefore its validity cannot be verified. However, the basic underlying reasoning of the proof is still valid, and with some modifications the original proof method can be salvaged --- we will now describe this in detail. This has been independently recognised~\cite{personal} by the authors of~\cite{faithful}. Furthermore, a closely related proof that avoids the use of the quantum Stein's lemma appeared subsequently in~\cite{rel-ent-sq}, and two alternative approaches to proving the faithfulness of squashed entanglement can be found in~\cite{VV-Markov} and~\cite{berta_2023}. 

The approach of~\cite{faithful} proceeds in three steps: namely, Lemmas~1,~2, and~3 in the paper. Lemma~1 relates the conditional mutual information with the regularised relative entropy of entanglement as
\begin{equation}\begin{aligned}
    I(A:B | E)_\rho \geq D^\infty_\SEP(\rho_{A:BE}) - D^\infty_\SEP(\rho_{A:E}).
\end{aligned}\end{equation}
Here and in the remainder of this section, we use $\rho_{A:B}$ to denote the bipartition with respect to which separability is considered. Lemma~2 of~\cite{faithful} then claims to give a monogamy--like relation for $D^\infty_\SEP$ as
\begin{equation}\begin{aligned}\label{eq:faithful_lemma2}
    D^\infty_\SEP(\rho_{A:BE}) - D^\infty_\SEP(\rho_{A:E}) \textgeq{?} \lim_{\epsilon \to 0} \lim_{n \to \infty} \frac1n \min_{\sigma_n \in \SEP} D^{\epsilon,\mathds{LOCC}_\leftarrow}_H \left(\rho^{\otimes n}_{A:B} \,\big\|\, \sigma_n \right),
\end{aligned}\end{equation}
where $D^{\epsilon,\mathds{E}}_H$ denotes the hypothesis testing relative entropy where the discriminating POVM $\{M, \id - M\}$ is restricted to the class of measurements $\mathds{E}$, and $\mathds{LOCC}_\leftarrow$ represents the measurements that are implementable with local operations assisted by one-way classical communication from Bob to Alice. Lemma~3 concludes the result by lower bounding the right-hand side of~\eqref{eq:faithful_lemma2} with a distance between $\rho_{A:B}$ and all separable states. Lemmas~1 and~3 are unaffected by the considerations of this note, but the proof of Lemma~2 directly uses the generalised quantum Stein's lemma to relate $D^\infty_\SEP$ with the hypothesis testing error exponent. To circumvent this problem, one can instead work directly with the exponent
\begin{equation}\begin{aligned}
   \lim_{\epsilon \to 0} \limc{\liminf_{n \to \infty}} \frac1n \min_{\sigma_n \in \SEP} D^\epsilon_{H} \left(\rho_{A:B}^{\otimes n} \,\big\|\, \sigma_n\right) =  E_D^{\mathrm{ANE}}(\rho_{A:B}).
\end{aligned}\end{equation}
We formalise this approach as follows.

\begin{lemma}\label{faithful_lemma1_new}
For any bipartite state $\rho_{AB}$ with extension $\rho_{ABE}$, it holds that
\begin{equation}\begin{aligned}\label{eq:faithful_lemma1_new}
    I(A:B | E)_\rho \geq E_D^{\mathrm{ANE}}(\rho_{A:BE}) - E_D^{\mathrm{ANE}}(\rho_{A:E}).
\end{aligned}\end{equation}
Furthermore, Lemma~2 of~\cite{faithful} shows that
\begin{equation}\begin{aligned}\label{eq:faithful_lemma2_new}
    E_D^{\mathrm{ANE}}(\rho_{A:BE}) - E_D^{\mathrm{ANE}}(\rho_{A:E}) \geq \lim_{\epsilon \to 0} \limc{\liminf_{n \to \infty}} \frac1n \min_{\sigma_n \in \SEP} D^{\epsilon,\mathds{LOCC}_\leftarrow}_H \left(\rho^{\otimes n}_{A:B} \,\big\|\, \sigma_n \right),
\end{aligned}\end{equation}
together with Lemma~3 of~\cite{faithful} implying the faithfulness of squashed entanglement $E_{\rm sq}(\rho_{A:B})$.
\end{lemma}
The relation in~\eqref{eq:faithful_lemma2_new} is already shown in~\cite{faithful} explicitly; one should simply ignore the statements that equate $E_D^{\mathrm{ANE}}$ with $D^\infty_\SEP$. More precisely, the proof of Lemma~2 on pp.~820--822 of~\cite{faithful} is overall valid because it proceeds by operational arguments, constructing a test for state discrimination on $A:BE$ from a $\mathds{LOCC}_\leftarrow$ test on $A:B$ 
and a generic test on $A:E$. The two identities that are no longer known to be correct are the unnumbered equation below Eq.~(66) and the last line of Eq.~(77) in~\cite{faithful}.

It thus remains to show~\eqref{eq:faithful_lemma1_new}. Inspecting the proof of Lemma~1 in~\cite{faithful}, one can notice that it crucially depends on a property of $D^\infty_\SEP$ known as non-lockability. Our approach will be to derive a similar relation for $E_D^{\mathrm{ANE}}$ directly. Combined with the known asymptotic results concerning achievable rates of quantum state redistribution~\cite{devetak_2008-1}, this will establish~\eqref{eq:faithful_lemma1_new} as desired.

Let us begin by recalling the setting of quantum state redistribution. Here, we begin with a pure state $\rho_{ABER}$ shared by three parties: the sender, the receiver, and a purifying referee. For the sake of consistency with the notation of~\cite{faithful}, let us assume that the sender holds the systems $RB$ and aims to send $B$ to the receiver, who is in possession of the system $E$. The two parties additionally share a supply of maximally entangled states in systems $T_1 T_2$, with $T_1$ belonging to the sender and $T_2$ to the receiver. The system $A$ represents the referee. Quantum state redistribution then proceeds in three steps:
\begin{enumerate}[(1)]
\item The sender applies a local encoding operation $RBT_1 \to RQT'_1$.
\item The sender sends the system $Q$ to the receiver.
\item The receiver applies a local decoding operation $QET_2 \to BET'_2$.
\end{enumerate}
 Let us then assume that there exists an $\ve$-error one-shot state redistribution protocol $\Lambda$, that is, that for some suitable encoding and decoding operations we can obtain a state $\omega_{A:BERT'_1T'_2} = \Lambda (\rho_{AB:ER} \otimes \Phi_{T_1T_2})$ such that
\begin{equation}\begin{aligned}
    P\left(\omega_{A:BERT'_1T'_2}, \rho_{A:BER} \otimes \Phi_{T'_1T'_2}\right) \leq \ve
\end{aligned}\end{equation}
where $P(\rho,\omega) = \sqrt{1-F(\rho,\omega)}$ is the purified distance, with $F(\rho,\omega) = \| \sqrt{\rho}\sqrt{\vphantom{\rho}\omega} \|^2$ the (squared) fidelity. Here we are interested only in the communication cost of this task, that is, the dimension of the system $Q$; we will denote this by
\begin{equation}\begin{aligned}
    q^\ve(B:E|A)_{\rho} \coloneqq \log |Q|.
\end{aligned}\end{equation}

Our main contribution is then the following result.

\begin{lemma}\label{lem:oneshot_redist}
The communication cost of any $\ve$-error one-shot state redistribution protocol as described above satisfies
\begin{equation}\begin{aligned}
   q^\ve(B:E|A)_{\rho} \geq \frac12 \min_{\sigma \in \SEP} D^{\ve}_H (\rho_{A:BE} \| \sigma) + \frac12 \min_{\sigma' \in \SEP} D^{3\ve}_H (\rho_{A:E} \| \sigma') - \log \frac{4}{\ve^2}.
\end{aligned}\end{equation}
\end{lemma}
\begin{proof}
For simplicity, we will employ the notation $D^\ve_H (\rho_{A:BE} \| \SEP) \coloneqq \min_{\sigma \in \SEP(A:BE)} D^\ve_H(\rho_{A:BE} \| \sigma)$ and analogously for $D^\ve_{\max}$. We will now fix the smoothing in the definition of $D^\ve_{\max}$ to be with respect to the purified distance, $\Delta(\rho,\rho') = P(\rho,\rho')$.

Pick some $\eta, \delta \in (0,1)$ such that $\eta - \delta \in (\ve, 1)$. Denoting by $\omega_{A:BERT'_1T'_2}$ a state such that $P\left(\omega_{A:BERT'_1T'_2}, \rho_{A:BER} \otimes \Phi_{T'_1T'_2}\right) \leq \ve$ as above, we have that
\begin{align}
    D^{\eta - \delta - \ve}_H (\rho_{A:BE} \| \SEP) - 2 \log \frac{2}{\delta}\ &\textleq{(i)}\ D^{\eta - \delta}_H (\omega_{A:BE} \| \SEP) - 2 \log \frac{2}{\delta} \nonumber \\
    &\textleq{(ii)}\  D^{\eta - \delta}_H (\omega_{A:BET'_2} \| \SEP) - 2 \log \frac{2}{\delta} \nonumber \\
    &\textleq{(iii)}\ D^{\eta - \delta}_H (\omega_{A:QET_2} \| \SEP) - 2 \log \frac{2}{\delta} \nonumber \\
    &\textleq{(iv)}\  D^{\sqrt{1-\eta}}_{\max} (\omega_{A:QET_2} \| \SEP) - \log \frac{1}{\eta} \nonumber \\
    &\textleq{(v)}\  D^{\sqrt{1-\eta}}_{\max} (\omega_{A:ET_2} \| \SEP) - \log \frac{1}{\eta} + 2 \log |Q| \\
    &\texteq{(vi)}\  D^{\sqrt{1-\eta}}_{\max} (\rho_{A:E} \otimes \Phi_{T_2} \| \SEP) - \log \frac{1}{\eta} + 2 \log |Q| \nonumber \\
    &\texteq{(vii)}\  D^{\sqrt{1-\eta}}_{\max} (\rho_{A:E} \| \SEP) - \log \frac{1}{\eta} + 2 \log |Q| \nonumber \\
    &\textleq{(viii)}\  D^{\eta}_{H} (\rho_{A:E} \| \SEP) + 2 \log |Q|. \nonumber
\end{align}
Here: (i)~follows because $P(\rho,\omega) \leq \ve \Rightarrow D^{\eta + \ve}_H(\omega \| \sigma) \geq D^{\eta}_H(\rho \| \sigma)$ (Lemma~\ref{lem:DH_continuity} in the Appendix); (ii)~is through the data processing inequality for $D^\ve_H$; (iii)~follows again by data processing --- this is step~(3) of the quantum state redistribution protocol; in~(iv) we use the inequality~\cite[Theorem~4]{anshu_2019}
\begin{equation}\begin{aligned}
     D^{\sqrt{1-\eta}}_{\max}(\rho \| \sigma) - \log\frac{1}{\eta} \geq D^{\eta-\delta}_H (\rho \| \sigma) - 2 \log\frac{2}{\delta};
 \end{aligned}\end{equation} 
(v)~is the non-lockability of the smooth max-relative entropy, which we show in Lemma~\ref{lem:Dmax_nonlock}; (vi)~is because steps~(1) and~(2) of the state redistribution protocol leave systems $AET_2$ unchanged; (vii)~follows because $D^\ve_{\max}(\cdot \| \SEP)$ does not change under tensor product with a separable state due to its monotonicity under all non-entangling maps 
(this is a special case of~\cite[Lemma~2]{brandao_2011}); finally, (viii)~follows by the inequality~\cite[Theorem~4]{anshu_2019}
\begin{equation}\begin{aligned}
     D^{\eta}_H (\rho \| \sigma) \geq  D^{\sqrt{1-\eta}}_{\max}(\rho \| \sigma)  - \log \frac{1}{\eta}.
 \end{aligned}\end{equation} 
 Choosing $\eta = 3\ve, \delta = \ve$ concludes the proof.
\end{proof}

\begin{proof}[Proof of Lemma~\ref{faithful_lemma1_new}]
We start by applying the result of Lemma~\ref{lem:oneshot_redist} to $\rho^{\otimes n}$. Dividing by $n$, taking the limits $n \to \infty$ and $\ve \to 0$, and using the fact that there exists a sequence of state redistribution protocols such that~\cite{devetak_2008-1}
\begin{equation}\begin{aligned}
    \lim_{\ve \to 0} \limc{\liminf_{n \to \infty}} \frac1n q^\ve(B:E|A)_{\rho^{\otimes n}} =  I(A:B | E)
\end{aligned}\end{equation}
gives the stated result.
\end{proof}


\subsection{Quantum channel manipulation}

The description of certain asymptotic properties of quantum channels, including generalisations of quantum hypothesis testing, often relies on generalised Stein--type results. This is because, even if one is interested in studying a channel $\Lambda$ in the i.i.d.\ setting as $\Lambda^{\otimes n}$, the fact that input states to the channel may be entangled means that $\Lambda^{\otimes n}(\rho)$ need not be i.i.d.\ for a general state $\rho \in \D(\HH^{\otimes n})$. This then necessitates the use of more general variants of the quantum Stein's lemma to fully understand the asymptotic behaviour of the involved channels and states, and in particular in the characterisation of strong converse properties of channel discrimination.

Results in the study of quantum channel manipulation that either directly employ Conjecture~\ref{gen_steins} (e.g.~\cite[Theorem~6]{Gour-Winter}) or make use of the same tools as~\cite{Brandao2010} (e.g.~\cite[Theorem~12]{fang_preprint}) can no longer be considered correct, and will require alternative proofs to verify their validity.
On the other hand, binary quantum channel discrimination is well understood at the level of weak converse rates~\cite{Berta2018,wang_2019-4,Fang2020}, where natural extensions of the quantum Stein's lemma have been shown independently of~\cite{Brandao2010}. The composite strong converse case has also been solved in some restricted settings, including in particular the quantum Shannon's theorem~\cite{Bennett2014,Bennett2002,Berta2011} or the related setting of distinguishing a channel from the set of replacer channels~\cite{Cooney2016}.

\section{Discussion}

We conclude by recounting several known approaches that could be used to tackle the newly reopened question of the generalised quantum Stein's lemma. Let us first recall that the main claim (Conjecture~\ref{gen_steins}) concerns the question of whether
\begin{equation}\begin{aligned}
     \lim_{\ve \to 0} \limc{\liminf_{n \to \infty}} \frac1n \min_{\sigma_n \in \MM_n} D_{H}^{\ve} (\rho^{\otimes n} \| \sigma_n) \texteq{?} D^\infty_\MM(\rho)
\end{aligned}\end{equation}
for a family of sets $(\MM_n)_n$ of quantum states in consideration, and in particular for the separable states $(\SEP_n)_n$. This problem can be understood in several ways.

\begin{itemize}
\item The direct achievability proof. For all states $\rho$, find a sequence of tests $(M_n)_n$ such that $0 \leq M_n \leq \id$, $\lim_{n\to\infty} \Tr \rho^{\otimes n} M_n  = 1$, and
\begin{equation}\begin{aligned}
    \limc{\liminf_{n \to \infty}} - \frac1n \sup_{\sigma_n \in \MM_n} \log \Tr \sigma_n M_n = D^\infty_\MM(\rho),
\end{aligned}\end{equation}
or show that there exists a state $\rho$ for which no such tests can exist.

\item The strong converse $D_{\max}$ approach. As discussed previously, the max-relative entropy satisfies the property that~\cite{Brandao2010,datta_2009-2}
\begin{equation}\begin{aligned}
    \lim_{\ve \to 0} \limc{\limsup_{n \to \infty}} \frac1n \min_{\sigma_n \in \MM_n} D_{\max}^{\ve} (\rho^{\otimes n} \| \sigma_n) =     \lim_{\ve \to 0} \limc{\liminf_{n \to \infty}} \frac1n \min_{\sigma_n \in \MM_n} D_{\max}^{\ve} (\rho^{\otimes n} \| \sigma_n) =  D^\infty_\MM(\rho).
\end{aligned}\end{equation}
It is furthermore known that $D^\ve_H$ is essentially equivalent to $D^{\sqrt{1-\ve}}_{\max}$~\cite{tomamichel_2013,anshu_2019}, in the sense that the two functions are equal up to $\ve$-dependent terms that vanish asymptotically. Because of this, the statement of Conjecture~\ref{gen_steins} is equivalent to the question of whether
\begin{equation}\begin{aligned}
    \lim_{\ve \to 0} \limc{\liminf_{n \to \infty}} \frac1n \min_{\sigma_n \in \MM_n} D_{\max}^{1-\ve} (\rho^{\otimes n} \| \sigma_n) \texteq{?} D^\infty_\MM(\rho).
\end{aligned}\end{equation}
Analogous statements can be obtained using other quantities that are asymptotically equivalent to $D^\ve_H$, for instance the information spectrum relative entropy~\cite{tomamichel_2013}.

\item Petz--R\'enyi limit interchange approach. The hypothesis testing relative entropy is lower bounded by the Petz--R\'enyi divergences~\cite{PetzRenyi} $D_\alpha(\rho\|\sigma) \coloneqq \frac{1}{\alpha-1} \log \Tr \rho^{\alpha} \sigma^{1-\alpha}$ as~\cite{Audenaert2008}
\begin{equation}\begin{aligned}
    D^\ve_H(\rho \| \sigma) \geq D_\alpha(\rho\|\sigma) + \frac{\alpha}{1-\alpha} \log \frac1\ve
\end{aligned}\end{equation}
for all $\alpha \in (0,1)$. Furthermore, it is well known that $D_\alpha$ converge to the quantum relative entropy in the limit as $\alpha\to 1$~\cite{PetzRenyi}. Because of this, showing that
 \begin{equation}\begin{aligned}\label{eq:petz}
     \lim_{\alpha\to1^-} \limc{\liminf_{n\to\infty}} \frac1n \min_{\sigma_n \in \MM_n} D_\alpha(\rho^{\otimes n} \| \sigma_n) \texteq{?}  \lim_{n\to\infty} \lim_{\alpha\to1^-}  \frac1n \min_{\sigma_n \in \MM_n} D_\alpha(\rho^{\otimes n} \| \sigma_n) = D^\infty_\MM(\rho)
 \end{aligned}\end{equation}
 would imply the generalised quantum Stein's lemma. The rightmost equality in~\eqref{eq:petz} follows by standard minimax arguments~\cite[Corollary~A.2]{mosonyi_2011}.
 \end{itemize}
 
For the case of entanglement theory (i.e.\ $\MM$ being the separable states), 
one can also equivalently study an operational argument based on entanglement distillation.

\begin{itemize}
\item For all states $\rho$ and all $\ve>0$, show that there exists an asymptotically non-entangling protocol that distills entanglement at the rate $r = D^\infty_S(\rho) - \ve$, or show that there exists a state $\rho$ whose rate of distillation under all asymptotically non-entangling protocols (equivalently: all non-entangling protocols, c.f.~\eqref{eq:distillable_ane}) is bounded away from $D^\infty_S(\rho)$.

\end{itemize}

Solving any of the above problems in the affirmative would restore the conjectured generalised quantum Stein's lemma, whose importance we hope to have outlined in this paper. This would in particular re-establish the claimed results of~\cite{BrandaoPlenio1,BrandaoPlenio2,Brandao-Gour}, showing the possibility of the existence of a reversible theory of entanglement and/or other quantum resources. A negative resolution of the conjecture, on the other hand, would carry a flurry of consequences of its own, notably reinvigorating the questions of whether general quantum resources can be reversibly manipulated whatsoever and what the asymptotic error exponent in discriminating a quantum state from a general set of quantum states is.


\begin{acknowledgments}
We are grateful to Ryuji Takagi for making us aware of an inaccuracy in our discussion in Section~V. We thank Alexander Streltsov and Mark M.\ Wilde for comments on the manuscript. 
MB acknowledges funding by the European Research Council (ERC Grant Agreement No.\ 948139). This project was started when MB was additionally affiliated with the AWS Center for Quantum Computing, Pasadena, USA. LL was partly financially supported by the Alexander von Humboldt Foundation. GG is supported by Natural Sciences and Engineering Research Council of Canada (NSERC). BR was partially supported by the Japan Society for the Promotion of Science (JSPS) KAKENHI Grant No.\ 22KF0067 and the JSPS Postdoctoral Fellowship for Research in Japan. MT is funded by the National Research Foundation, Prime Minister's Office, Singapore and the Ministry of Education, Singapore under the Research Centres of Excellence programme, as well as startup grants (R-263-000-E32-133 and R-263-000-E32-731).
\end{acknowledgments}


\bibliography{biblio}


\appendix

\section{Auxiliary lemmas}

\begin{lemma}\label{lem:DH_continuity}
Let $\omega$ be a state  such that $P(\rho,\omega)\leq \delta$. Then, for any state $\sigma$, $D_H^{\ve+\delta} (\rho \| \sigma) \geq D_H^\ve(\omega \| \sigma)$.
\end{lemma}
\begin{proof}
By the Fuchs--van de Graaf inequality~\cite{Fuchs1999}, we have that $\frac12 \|\rho - \omega\|_1 \leq \delta$. By definition of $D^\ve_H$, there exists a measurement operator $M \in [0,\id]$ such that $\Tr \left[(\id - M) \omega\right] \leq \ve$. Then
\begin{equation}\begin{aligned}
    \Tr \left[ (\id - M) \rho \right] &= \Tr \left[ (\id - M) (\rho - \omega)\right] + \Tr \left[ (\id - M) \omega \right]\\
    &\leq \frac12 \|\rho - \omega\|_1 + \ve\\
    &\leq \delta + \ve,
\end{aligned}\end{equation}
where in the second line we used the fact that
\begin{equation}\begin{aligned}
    \frac12 \|\rho - \omega\|_1 = \max_{0 \leq W \leq \id} \Tr \left[ W (\rho - \omega) \right]
\end{aligned}\end{equation}
for normalised states $\rho, \omega$. This means that $M$ is a feasible measurement operator for $D^{\ve+\delta}_H(\rho \| \sigma)$ as desired.
\end{proof}

\begin{lemma}\label{lem:Dmax_nonlock}
For any tripartite state $\rho_{ABE}$, it holds that
\begin{equation}\begin{aligned}
     D^\ve_{\max}(\rho_{A:BE} \| \SEP) \leq D^\ve_{\max}(\rho_{A:E} \| \SEP) + 2 \log |B|.
\end{aligned}\end{equation}
\end{lemma}
\begin{proof}
Let $\sigma'_{A:E}$ be an optimal separable state and $\widetilde\rho_{AE}$ an optimal state with $P(\rho_{AE},\widetilde\rho_{AE}) \leq \ve$ such that $\widetilde\rho_{A:E} \leq \lambda \sigma'$ and $ D^\ve_{\max}(\rho_{A:E} \| \SEP) = 
\log\lambda$. Then there exists an extension $\widetilde\rho_{ABE}$ of $\widetilde\rho_{AE}$ such that $P(\rho_{ABE},\widetilde\rho_{ABE})\leq \ve$~\cite[Corollary~3.14]{TOMAMICHEL}. For any such extension, we have that
\begin{equation}\begin{aligned}
    \widetilde\rho_{A:BE} \leq |B| \id_{B} \otimes \widetilde\rho_{A:E} \leq \lambda |B|^2 \frac{\id_B}{|B|} \otimes \sigma',
\end{aligned}\end{equation}
where the first inequality is by the complete positivity of the map $\rho_{B} \mapsto \Tr(\rho_{B}) |B| \id_{B} - \rho_{B}$, easily verified by observing that the Choi operator of this map is $|B| \id_{BB} - |B| \Phi_{BB} \geq 0$. Since the state $ \frac{\id_B}{|B|} \otimes \sigma'$ is separable across $A:BE$, we have that $D^\ve_{\max}(\rho_{A:BE} \| \SEP) \leq \log (\lambda |B|^2)$.
\end{proof}

\end{document}